\documentclass[%
 reprint,
superscriptaddress,
 amsmath,amssymb,
 aps,
]{revtex4-2}

\usepackage{graphicx}
\usepackage{dcolumn}
\usepackage{bm}
\usepackage{color}
\usepackage{amsfonts}
\usepackage{dsfont}
\usepackage{amsmath}
\usepackage{amssymb}
\usepackage{changes}
\usepackage{comment}
\usepackage{wrapfig}
\usepackage{hyperref}
\usepackage{physics}
\usepackage{times}
\usepackage{newtxmath}
\usepackage{mathtools}
\usepackage{tcolorbox}

\newtheorem{theorem}{Theorem}
\newtheorem{lemma}{Lemma}
\newtheorem{proposition}{Proposition}

\newtheorem{definition}{Definition}

\newtheorem{corollary}{Corollary}

\newenvironment{proof}[1][Proof]{\noindent\textbf{#1.} }{\ \rule{0.5em}{0.5em}}

 \newcommand{\proj}[1]{\ket{#1}\bra{#1}}

\newcommand{\loc}{{\rm L}}

\definecolor{nred}{rgb}{0.9,0.1,0.1}
\definecolor{pur}{rgb}{0.75,0,0.75}

\newcommand{\paolo}[1]{{\color{orange} #1}}

\begin{document}


\title{
Quantum steering in networks: \\Measurement-device-independent detection, continuous variables, and practical Gaussian schemes
%
}

\author{Shao-Hua Hu\hyperlink{equal}{\textsuperscript{$\star$}}}\email{shhphy@gmail.com}
\affiliation{Department of Physics, National Tsing Hua University, Hsinchu 30013, Taiwan}

\author{Chung-Yun Hsieh}
\affiliation{H. H. Wills Physics Laboratory, University of Bristol, Tyndall Avenue, Bristol, BS8 1TL, United Kingdom}

\author{Huan-Yu Ku}
\affiliation{Department of Physics, National Taiwan Normal University, Taipei 11677, Taiwan}

\author{Ray-Kuang Lee}\email{rklee@ee.nthu.edu.tw}
\affiliation{Department of Physics, National Tsing Hua University, Hsinchu 30013, Taiwan}
\affiliation{Institute of Photonics Technologies, National Tsing Hua University, Hsinchu 30013, Taiwan}
\affiliation{Department of Electrical Engineering, National Tsing Hua University, Hsinchu 30013, Taiwan}

\author{Paolo Abiuso\hyperlink{equal}{\textsuperscript{$\star$}}}\email{paolo.abiuso@oeaw.ac.at}
\affiliation{Institute for Quantum Optics and Quantum Information - IQOQI Vienna,
Austrian Academy of Sciences, Boltzmanngasse 3, A-1090 Vienna, Austria}


\begin{abstract}
%
We consider quantum steering certification in multipartite networks, with a focus on \emph{minimal trust scenarios}: all-except-one parties are untrusted and treated device-independently.
We show that it is always possible to lift steering certification to the measurement-device-independent regime, in which even the (last) trusted party can treat their local hardware as a black-box, except for a set of fiduciary quantum states used as the inputs to the experiment. 
This holds both for finite-dimensional systems as well as for bosonic continuous-variable systems, which we characterize in terms of semi-quantum games in the bipartite case.
Additionally, we introduce measurement-device-independent network steering protocols based entirely on Gaussian operations -- which cannot be used for fully device-independent protocols, and thus become instead a viable option for minimal trust certification as soon as a single trusted input is inserted in the network.
Our results present a basis for steering-based applications (such as randomness generation) with minimal trust beyond full nonlocality and with feasible experimental requirements.

\end{abstract}

\maketitle


\newpage

\section{Introduction}

Large-scale quantum networks
require nonlocal resources as a central ingredient for achieving quantum advantage in  communication tasks~\cite{GisinRMP2002,Chiribella2009,Wehner2018QuantumInternet,Azuma2023QuantumRepeaters,Main2025Distributed}. 
To this end, many protocols are proposed to distribute nonlocal resources in a network, e.g., by generating a cluster state~\cite{Raussendorf2001,Yang2022,Roh2025,GashuFeyisa2025QST}, 
distributing entanglement
~\cite{Neumann2022,Jiang2026,Zhang2026PRL}, and 
routing quantum nodes~\cite{ChudzickiPRL2010,Hahn2019,Kristjnsson2024npjQI}. While several nonlocal quantum communication protocols, such as teleportation and key distribution, have been successfully implemented in networks~\cite{RossetPRL2016,Juan2017,Liao2017,MiguelPRL2020,VillegasAguilar2024}, 
their implementation in fully device-independent (DI) frameworks --such as in Bell nonlocality~\cite{BrunnerRMP2014,Tavakoli_2022}-- is notoriously hard in its experimental requirements, making such protocols usually quite experimentally demanding and even practically challenging.


These challenges grow in complexity when certifying entanglement and nonlocality detection in networks. 
This structural complexity stands in contrast to the simplest bipartite networks, where the problem reduces to standard entanglement or Bell nonlocality detection, with the latter underpinning fully DI protocols that certify quantum properties solely from the statistics of a causally constrained experiment. 
In the bipartite setting, various applications have already been extensively investigated~\cite{Cirac1997QSTandEntanglementPRL,AcinPRL2007,GUHNE20091,BrunnerRMP2014,VaziraniPRL2014,Chen2021robustselftestingof,Lucas2025PRR}. 
Beyond this, quantum networks also exhibit a much richer structure when considering intermediate schemes in which some nodes are trusted while others are not, i.e., \emph{semi-device-independent} protocols. In this configuration, detection of nonlocal correlations can be formulated as network quantum steering~\cite{Armstrong2015,Cavalcanti2015,JonesPRL2021,Zhang2023,Porto2026arXiv,Sarkar2026PRA,Sarkar2024Quantum}. Analogous to the bipartite setting, this framework describes the ability of an uncharacterized party to remotely change, or ``steer," the quantum state of a trusted node through local measurements. 
While being a notion of nonclassicality of historical importance in the early development of quantum foundations~\cite{WisemanPRL2007,CavalcantiSteeringRev2017,UolaRMP2020,XiangPRXQ2022}, the operational features of quantum steering now
lie at the core of one-sided device-independent (1S-DI) quantum information processing, such as in randomness certification~\cite{Passaro2015optimal,joch2022certified,ioannou2022steeringbased,sarkar2023self,li2024randomness,li2025necessary,huang2025experimental}.
Steering is subject to intense study in the field of quantum information, for its links to measurement incompatibility~\cite{quintino2014joint,UolaPRL2015,Ku2022NC,ku2022quantifying,Hsieh2023arXiv,Ku2024arXiv}, instrument incompatibility~\cite{Ji2024PRXQ,Hsieh2024PRL}, and thermodynamic work extraction~\cite{Hsieh2024PRL,Hsieh2025arXiv,Beyer2019PRL,Ji2022PRL,Biswas2025PRL}.

In this context, measurement-device-independent (MDI) resource detection relaxes the stringent requirement of full device independence by focusing on scenarios where only (some of) the state-preparation devices are locally trusted. The MDI framework has been successfully applied to detect various bipartite resources, either by constructing semi-quantum games or by reformulating resource witnesses, including entanglement~\cite{buscemi2012all,branciard2013measurement,XuPRL2014,VerbanisPRL2016,LiPRL2020}, teleportation, and their dynamical variants~\cite{CavalcantiPRL2017,rosset2018resource,IvanPRA2019,PatrykPRR2020,PatrykPRXQ2021}. Consequently, those resources are detectable via Bell-like inequalities with quantum inputs, bypassing the need for full state/device characterization at the trusted side. In Refs.~\cite{abiuso2021measurement,abiuso2023verification,Larsen2026continuous}, the MDI witnessing of different resources was promoted to the regime of continuous variables (CV), which is crucial to applications in quantum technologies~\cite{braunstein2005quantum,weedbrook2012gaussian}.
In the case of standard bipartite quantum steering, its MDI certification has been put forward in a series of works~\cite{cavalcanti2013entanglement,zhao2020experimental,guo2019measurement,ZhaoOptica23} for discrete-variable systems only.
Moreover, the MDI extension to network scenarios remains significantly limited.

In this work, we overcome previous restrictions and extend the MDI detection of steering to the network scenario, and including the case of CV systems. Specifically, we investigate the certification of quantum steering in multipartite networks under minimal-trust scenarios, where \emph{all but one} party are treated device-independently. We show that any network steering detection in these minimal-trust scenarios can be systematically lifted to the MDI regime. Within this framework, even the trusted party can treat their measurement device as a black-box, given access to fiducial quantum input states, thereby removing the need for, e.g., measurement calibration. 
In the bosonic CV case, we establish, for the first time, a comprehensive 
framework for CV semi-quantum games in the bipartite setting. Building on this, we introduce MDI network steering protocols that rely solely on Gaussian operations. Although such operations are known to be insufficient for fully DI tasks, we demonstrate that all Gaussian steerable states can be detected in an
experimentally-friendly MDI scheme.
We conclude that MDI steering scenarios thus represent the minimal relaxation of fully DI protocols, namely with a \emph{single trusted input} in the network, which is sufficient to enable resource certification with fully Gaussian technology.

\section{Network Steering and its MDI Detection}

\subsection{{Quantum steering in networks}}
While originally observed as a feature for bipartite quantum correlations~\cite{wiseman2007steering,jones2007entanglement}, generalized notions of steering in multipartite quantum networks have been considered in recent works~\cite{jones2021network,li2025detecting,fan2025quantum}. 
In particular, in Ref.~\cite{jones2021network}, network quantum steering is introduced based on the idea of \emph{network-local-hidden-state} (NLHS) models, for semi-device-independent network experiments involving some ($m$) of the ($n$) parties being trusted, while the remaining ($n-m$) are not. Intending to focus on minimal-trust scenarios, we will consider, in the following, the case $m=1$, both to keep the presentation simple as well as to keep a spirit analogous to standard bipartite steering. Notice, however, that our results can be lifted to the 
$m>1$
case.

We thus consider a network of $n$ parties, consisting of untrusted (black-box) parties  denoted by $A_i$ ($i=1,\dots,n-1$) with local input-output $x_i$ and $a_i$, and a single trusted party Bob (denoted as $B$). 
Collectively, we write $\vec{x}=(x_1,...,x_{n-1})$ and $\vec{a}=(a_1,...,a_{n-1})$.
As all-except-one parties are treated device-independently, the entire information of the experiment is contained in the reduced state assemblage in $B$, denoted by $\tilde{\sigma}_B^{\vec{a}|\vec{x}}\ge0$, 
which represents the unnormalized quantum state owned by $B$ resulting from the inputs $\vec{x}$ and outputs $\vec{a}$ of $\{A_i\}_{i=1}^{n-1}$.
More precisely, when the inputs $\vec{x}$ are selected and outputs $\vec{a}$ are produced, Bob obtains the normalized state 
\begin{align}
\frac{\tilde{\sigma}_B^{\vec{a}|\vec{x}}}{p(\vec{a}|\vec{x})}
\quad \text{with  probability}\quad 
\label{Eq:p(a|x)}
p(\vec{a}|\vec{x})\coloneqq\Tr[\tilde{\sigma}_B^{\vec{a}|\vec{x}}].
\end{align}
Importantly, as in every steering experiment, the inputs $\vec{x}$ cannot change the quantum outcome without post-selection (i.e., the classical-to-quantum no-signalling).
This is mathematically expressed as
\begin{align}
\sum_{\vec{a}}\tilde{\sigma}_B^{\vec{a}|\vec{x}}=\sigma_B \quad\forall\vec{x},
\end{align}
where $\sigma_B$ is a normalized state independent of $\vec{x}$.
Formally, the collection $\{\tilde{\sigma}_B^{\vec{a}|\vec{x}}\}_{\vec{a},\vec{x}}$ is called an {\em state assemblage}, or simply {\em assemblage}.


\begin{figure}
    \centering
    \includegraphics[width=\linewidth]{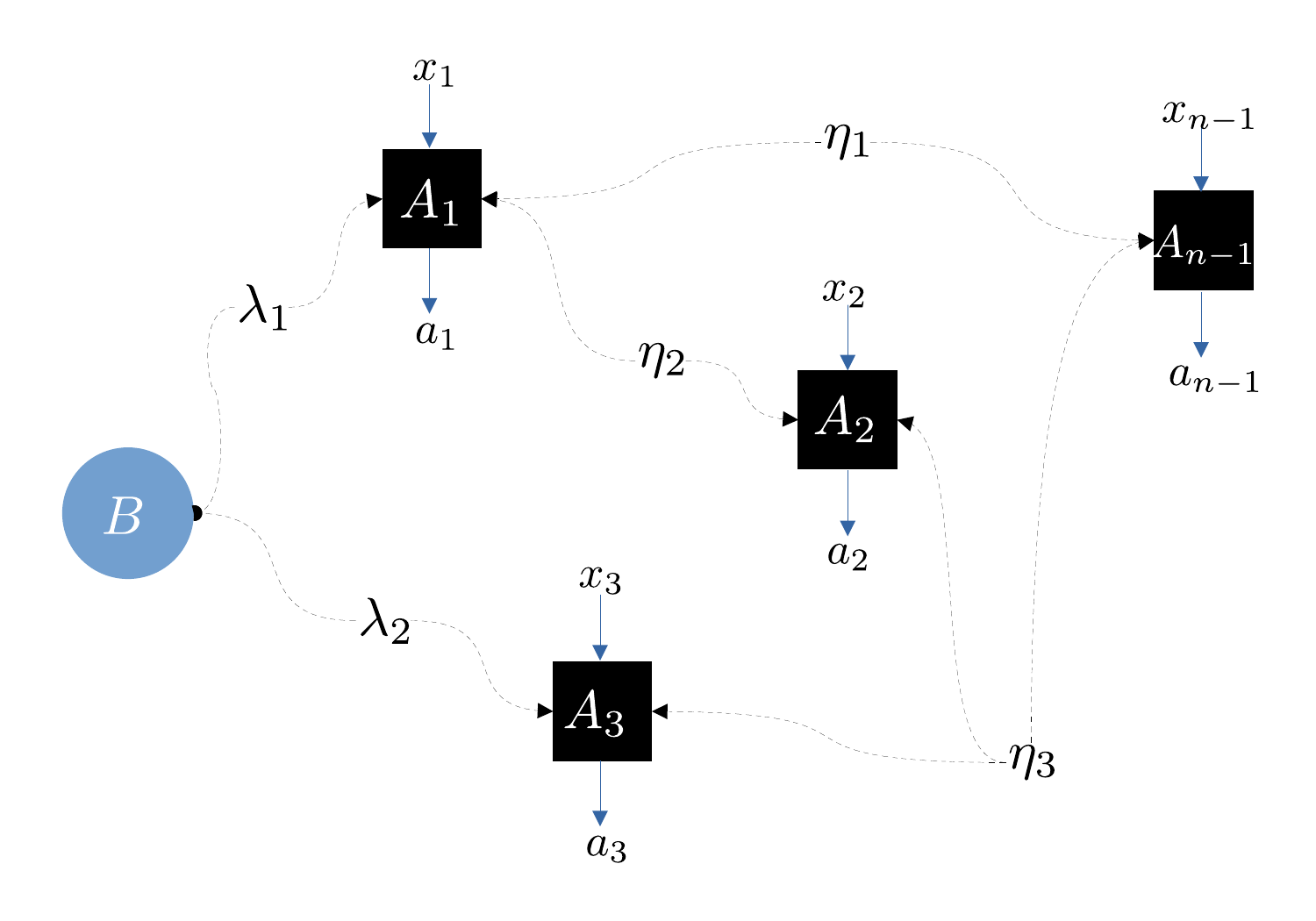}
    \caption{{\bf Network unsteerable assemblages} are defined by admitting a network-local-hidden-state (NLHS) model (here pictured for a single trusted party): a network of $n-1$ black-box parties $A_i$ and one trusted $B$ perform an experiment arranged in a given no-signalling network, resulting in the assemblage $\tilde{\sigma}_B^{\vec{a}|\vec{x}}$ describing both the statistics of the black-box $A_i$'s and the corresponding quantum state in $B$. In an NLHS model, the hidden (classical) variables $\lambda_i$'s are shared among some of the $A_i$'s and $B$, and the hidden variables $\eta_i$'s are shared among different $A_i$'s only. The resulting assemblage can then be described by Eq.~\eqref{eq:NLHS}. Experiments using quantum states instead of the classical variables $\{\vec{\lambda},\vec{\eta}\}$ can violate the decomposition in Eq.~\eqref{eq:NLHS} and are then said to feature network steering.}
    \label{fig:net_steer_sketch}
\end{figure}

An NLHS model for an assemblage in this scenario corresponds to admitting a decomposition compatible with the distribution of unobserved  \emph{classical} variables (also known as {\em hidden variables}) in the given network.
Physically, this means Bob's assemblage can be prepared without sharing any quantum correlation with any $A_i$'s at all.
Mathematically, an assemblage $\tilde{\sigma}_B^{\vec{a}|\vec{x}}$ is said to admit such an NLHS model if 
\begin{align}
\label{eq:NLHS}
    \tilde{\sigma}_B^{\vec{a}|\vec{x}}&=\sum_{\vec{\lambda}} \sigma^{(\vec{\lambda})}_B p_{\loc}(\vec{a}|\vec{x},\vec{\lambda})\mu(\vec{\lambda})\;,
\\ \nonumber
    \vec{\lambda}&=(\lambda_{1},\lambda_{2},\dots,\lambda_{l_B})\;, \quad \mu(\vec{\lambda}){\coloneqq}\prod_i^{l_B}\mu_i(\lambda_i).
\end{align}
Here, $\vec{\lambda}$ is the vector representing $l_B$ hidden variables $\lambda_1,...,\lambda_{l_B}$.
Each $\lambda_i$ is shared between $B$ and a subset of untrusted parties, 
with an independent distribution $\mu_i$.
See Fig.~\ref{fig:net_steer_sketch} for a schematic illustration. 
Also, each $\sigma^{(\vec{\lambda})}_B$ is a normalized state (namely, $\sigma^{(\vec{\lambda})}_B$ is positive and with unit trace). 
Finally, each conditional probability distribution $p_\loc(\vec{a}|\vec{x},\vec{\lambda})$ is {\em network-local} according to the network scenario restricted to the $A_i$'s only, namely, with a decomposition of the following form
for each value of $\vec{\lambda}$:
\begin{align}
\label{eq:local_dist}
    & p_{\loc}(\vec{a}|\vec{x},\vec{\lambda})=\\
    & \sum_{\vec{\eta}} p(a_1|x_1,\vec{\eta}^{(1)},\vec{\lambda}^{(1)}) \dots  p(a_{n-1}|x_{n-1},\vec{\eta}^{(n-1)},\vec{\lambda}^{(n-1)})\nu(\vec{\eta})\;,
    \nonumber
\\
    & \vec{\eta}=(\eta_{1},\eta_{2},\dots,\eta_{l_A})\;, \quad \nu(\vec{\eta}){\coloneqq}\prod_i^{l_A}\nu_i(\eta_i).
    \nonumber
\end{align}
Here, $\vec{\eta}$ is
a set of $l_A$ hidden variables $\eta_1,...,\eta_{l_A}$ shared only among the untrusted parties $A_i$'s, and each $\eta_i$ is sampled according to an independent probability distribution $\nu_i$.
Moreover, each $\vec{\eta}^{(i)}$ ($\vec{\lambda}^{(i)}$) is a subset of $\vec{\eta}$ ($\vec{\lambda}$) reaching $A_i$, as in the context of network nonlocality~\cite{Tavakoli2022Bell}.
That is, more formally, $\vec{\eta}^{(i)}=\{\eta_{k}\,|\,\eta_{k}\;\text{is shared by $A_i$}\}$, and similarly for $\vec{\lambda}^{(i)}$.

When the decomposition in Eq.~\eqref{eq:NLHS} does not exist, the resulting assemblage in the experiment presents network steering~\cite{jones2021network}, which reduces to the standard notion of (bipartite) steering for $n=2$.
Formally, we have the following definition:
\begin{definition}[Network quantum steering]
    An assemblage $\tilde{\sigma}_B^{\vec{a}|\vec{x}}$ that cannot be decomposed as from Eq.~\eqref{eq:NLHS} is said to be steerable in the assigned network.
\end{definition}

A few observations are in order. 
First, notice that the decomposition in Eq.~\eqref{eq:NLHS} fails already if the marginal distribution $p(\vec{a}|\vec{x})$ is nonlocal in the network obtained by removing $B$, or if any of the semi-marginalized states on some subset \mbox{$\mathcal{S}\subset\{1,\dots,n-1\}$} of the untrusted parties, {that is,} \mbox{$\sum_{\vec{a}|_{\bar{\mathcal{S}}}} \sigma^{\vec{a}|\vec{x}}_B{\eqqcolon}\sigma^{\vec{a}|\vec{x}}_B(\mathcal{S})$,} is steerable in the standard bipartite sense (namely, considering all $A_i$'s in $\mathcal{S}$ to be a joint, single party).
Further, in the more general case with $m\neq 1$~\cite{jones2021network}, entanglement among different trusted parties, denoted by $B_i$'s, makes the NLHS model fail as well. 
We keep this as additional motivation to restrain ourselves to $m=1$ in this work.

\subsection{MDI detection of all network-steerable assemblages 
}
Crucially, in a network steering setting, Bob's local hardware is fully trusted in order to tomographycally assess $\tilde{\sigma}_B^{\vec{a}|\vec{x}}$.
It is then useful to know whether we can further reduce such assumption on Bob's side, and
ask whether trusting Bob's measurement device is really necessary; namely:
\begin{center}
{\em
Is it possible to upgrade the detection of any network-steerable assemblage to an MDI scenario? 
}
\end{center}
This question has been answered, in the positive, for entanglement detection (both finite dimensional and infinite dimensional~\cite{buscemi2012all,branciard2013measurement,abiuso2021measurement}), as well as for finite-dimensional bipartite steering~\cite{cavalcanti2013entanglement,zhao2020experimental}, see Table~\ref{tab:context} for a visual summary and context. Here, we prove that the same is true in full generality for steering, namely,
\begin{enumerate}
\item[(i)] it is possible in any network scenario;
\item[(ii)] it is possible regardless of the underlying Hilbert space dimension.
\end{enumerate}
In the following, we provide a simple argument and its proof in the discrete-variable case, sparing the CV proof to the technical material in Appendix~\ref{app:steerproof}.

\begin{figure}
    \centering
    \includegraphics[width=\linewidth]{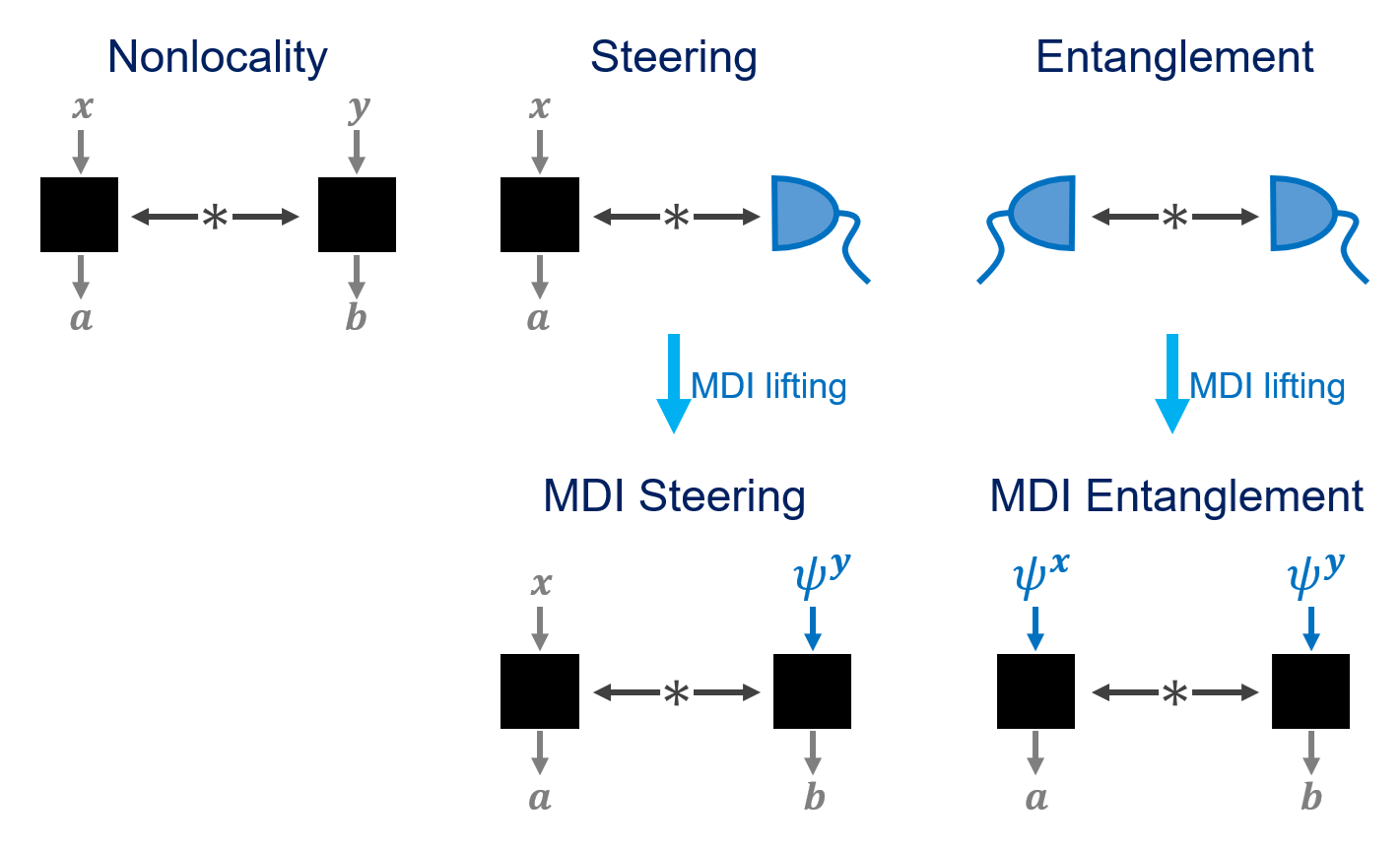}
    \begin{tabular}{c|c|c}
     &\quad\quad  MDI Steering \quad\ &\quad MDI Entanglement \\ \hline
     Discrete variable    & \cite{cavalcanti2013entanglement,zhao2020experimental} & \cite{buscemi2012all,branciard2013measurement} \\
     Continuous variable    & This work & \cite{abiuso2021measurement}
    \end{tabular}
    \caption{{\bf Hierarchy of possible device-trust in the bipartite scenario.}  
    \emph{i)} Nonlocality, \emph{ii)} steering, and \emph{iii)} entanglement correspond to witnessing correlations that cannot be explained by a classical shared variable, respectively, in experiments in which \emph{i)} no party is trusted, \emph{ii)} one party is trusted, \emph{iii)} both parties are trusted.  The figure shows scenarios where a quantum source emits, e.g., a pair of photons, which are sent to two distant agents. To certify quantumness, for trusted parties, we can simply trust data from their photon detectors. 
    For untrusted parties, which are viewed as back-boxes, one can only collect the resulting statistics (with local inputs $x,y$ and outputs $a,b$). Steering and entanglement witnesses can then be lifted to measurement-device-independent (MDI) protocols in which the trusted hardware is reduced to quantum inputs only (i.e., by replacing the trusted detectors by black-boxes with trusted quantum inputs $\psi^x,\psi^y$). Notice that by definition, there is no MDI lifting for nonlocality (as there is no trusted hardware).}.
    \label{tab:context}\label{fig:bipartite table}
\end{figure}


The underlying intuition is close in spirit to the derivations (for MDI entanglement witnessing) of Refs.~\cite{branciard2013measurement,vsupic2017measurement,abiuso2021measurement}:
starting from the network steering scenario described above, consider the case in which Bob's measurement station is also taken as a black-box, while only assuming a set of possible trusted quantum inputs (denoted by $\psi^y$) and the validity of quantum mechanics. This means that the statistics of the experiment are represented by
\begin{align}
    p(\vec{a},b|\vec{x},\psi^y)=\Tr\left[{M^{b}_{BB'} (\tilde{\sigma}_B^{\vec{a}|\vec{x}}\otimes\psi^y_{B'}})\right]\;,
    \label{eq:mdi_exp}
\end{align}
where $M_{BB'}^b$ is a positive operator-valued measure (POVM), which describes a local measurement in the system $BB'$ with outcome $b$.
Here, both the assemblage $\tilde{\sigma}_B^{\vec{a}|\vec{x}}$ and the local measurement $M_{BB'}^b$ are {\em unknown}; that is, we do not assume any trust in Bob's measurement device.
Eq.~\eqref{eq:mdi_exp} then defines an effective measurement on the inputs $\psi^{y}$ via
\begin{align}\label{Eq: MDI steering Bob's local stat}
p(\vec{a},b|\vec{x},\psi^y){\eqqcolon}\Tr[\mathcal{M}^{b,\vec{a}|\vec{x}}_{B'} \psi^y_{B'}]\;,
\end{align}  
which is given by
\begin{align}
    \mathcal{M}_{B'}^{b,{\vec{a}|\vec{x}}}\coloneqq\Tr_B\left[M^{b}_{BB'} {\left(\tilde{\sigma}_B^{\vec{a}|\vec{x}}\otimes\mathbb{I}_{B'}\right)}\right]\;.
    \label{eq:eff_meas}
\end{align}
Hence, by varying the inputs $\psi^y$, one can potentially use the above statistics to completely determine $\mathcal{M}_{B'}^{b,{\vec{a}|\vec{x}}}$.

Now, a crucial observation is: in the case where $\tilde{\sigma}_B^{\vec{a}|\vec{x}}$ is network unsteerable [Eq.~\eqref{eq:NLHS}], each POVM element $ \mathcal{M}_{B'}^{b,{\vec{a}|\vec{x}}}$ (whose measurement outputs are labelled by the index $b$) will have a decomposition akin to Eq.~\eqref{eq:NLHS}. 
In fact,
\begin{align}
\nonumber
    {\mathcal{M}_{B'}^{b,{\vec{a}|\vec{x}}}}&=\Tr_B\left[M^{b}_{BB'} \left(\tilde{\sigma}_B^{\vec{a}|\vec{x}}\otimes\mathbb{I}_{B'}\right)\right]\nonumber\\
    &=\sum_{\vec{\lambda}} \Tr_B\left[M^{b}_{BB'} \left(\sigma_B^{(\vec{\lambda})}\otimes\mathbb{I}_{B'}\right)\right] p_\loc(\vec{a}|\vec{x},\vec{\lambda})\mu(\vec{\lambda})\;\nonumber\\
    &\eqqcolon\sum_{\vec{\lambda}} \mathcal{M}^{b|\vec{\lambda}}_{B'} p_\loc(\vec{a}|\vec{x},\vec{\lambda})\mu(\vec{\lambda}),
\label{eq:povm_NLHS}
\end{align}
where we have defined $\mathcal{M}^{b|\vec{\lambda}}_{B'}\coloneqq\Tr_B\left[M^{b}_{BB'}\left(\sigma_B^{(\vec{\lambda})}\otimes\mathbb{I}_{B'}\right)\right]$, which 
induce an assemblage by fixing the value of $b$.
To see this, first, note that
the only difference between Eqs.~\eqref{eq:NLHS} and~\eqref{eq:povm_NLHS} is that, for a fixed $b$, the normalization $\sum_{\vec{a}}\Tr[ \mathcal{M}_{B'}^{b,{\vec{a}|\vec{x}}}]=\Tr[{M}_{BB'}^{b}\left(\sigma_B\otimes\mathbb{I}_{B'}\right)]\neq 1$ is in general different. However, notice that such a normalization can be reconstructed by
the statistics in Eq.~\eqref{Eq: MDI steering Bob's local stat} via a suitably chosen set of states $\psi^y$~
\footnote{
More precisely, by selecting an orthonormal basis $\psi^y$, this normalization factor reads
$
\Tr[{M}_{BB'}^{b}\left(\sigma_B\otimes\mathbb{I}_{B'}\right)]=\sum_{\vec{a},y}p(\vec{a},b|\vec{x},\psi^y),
$
which again can be obtained without trusting Bob's measurement device.}.
This thus means that:
\begin{proposition} 
\label{prop:meas_unsteer}
For a given assemblage $\tilde{\sigma}_B^{\vec{a}|\vec{x}}$ in a network steering experiment, the measurement $M^{b}_{BB'}$  defines an associated assemblage
\begin{align}
\tilde{m}_{B'}^{b,\vec{a}|\vec{x}}:= \begin{cases}\frac{\mathcal{M}_{B'}^{b,\vec{a}|\vec{x}}}{\Tr[{M}_{BB'}^{b}\left(\sigma_B\otimes\mathbb{I}_{B'}\right)]} & {\text{if}}\;\Tr[{M}_{BB'}^{b}(\sigma_B\otimes\mathbb{I}_{B'})]\neq 0\;,\\
0 & {\text{if}}\; \Tr[{M}_{BB'}^{b}(\sigma_B\otimes\mathbb{I}_{B'})]= 0.
\end{cases}
\label{eq:box_unsteer}
\end{align}
If the assemblage $\tilde{\sigma}_B^{\vec{a}|\vec{x}}$ is network unsteerable, then, for each value of $b$, $\tilde{m}_{B'}^{b,\vec{a}|\vec{x}}$ is unsteerable in the same network.
\end{proposition}

A crucial role in the shaping of $\mathcal{M}_{B'}^{b,\vec{a}|\vec{x}}$ in Eq.~\eqref{eq:eff_meas} and the corresponding $\tilde{m}_{B'}^{b,\vec{a}|\vec{x}}$ is played by Bob's measurement $M^{b}_{BB'}$. In particular, in case the latter is local, the possibility of MDI-steering must rely on the nonlocality of $p(\vec{a}|\vec{x})$ defined in Eq.~\eqref{Eq:p(a|x)} alone. 
In fact, consider $M_{BB'}^b$ to be carried out locally, aided by classical shared  randomness $\lambda$ and classical postprocessing $f(b|c_1,c_2,\lambda)$, 
\begin{align}
\label{eq:loc_meas}
    M_{BB'}^{b}=\sum_{c_1,c_2,\lambda}\mu(\lambda) f(b|c_1,c_2,\lambda) M_B^{c_1|\lambda}\otimes M_{B'}^{c_2|\lambda}\;,
\end{align}
it then follows that
\begin{proposition}
\label{prop:loc_meas_unsteer}
If Bob only applies local measurements of the form~\eqref{eq:loc_meas}, then the resulting 
$\tilde{m}_{B'}^{b,\vec{a}|\vec{x}}$ is unsteerable in the bipartition $\vec{A}|B$. As such, any MDI experiment can always be explained via an underlying bipartite unsteerable assemblage.
\end{proposition}
The proposition follows from the decomposition
\begin{align}
    \Tr_B\left[M^{b}_{BB'} {\left(\tilde{\sigma}_B^{\vec{a}|\vec{x}}\otimes\mathbb{I}_{B'}\right)}\right]
    &= \sum_{c_1,\lambda} \mu(\lambda)
    {p(c_1,\vec{a}|\vec{x},\lambda)}\mathcal{M}_{B'}^{b|c_1,\lambda}\;
\end{align}
arising from Eq.~\eqref{eq:loc_meas} with
\begin{align}
\nonumber
    p(c_1,\vec{a}|\vec{x},\lambda) &\coloneqq\Tr[M_{B}^{c_1|\lambda}\tilde{\sigma}_{B}^{\vec{a}|\vec{x}}]=p(c_1|\lambda)p(\vec{a}|\vec{x},c_1,\lambda)\;,
\\ 
    \mathcal{M}_{B'}^{b|c_1,\lambda} &\coloneqq\sum_{c_2} f(b|c_1,c_2,\lambda) M_{B'}^{c_2|\lambda}\;.
\end{align}
Proposition~\ref{prop:loc_meas_unsteer} shows that separable measurements prevent from detecting steering in an MDI scenario, implying that any MDI-steering detection must arise from an underlying entangled measurements in Bob's lab. 

Crucially, by using entangled measurements, we can show that a counterpart of Proposition~\ref{prop:meas_unsteer} is also possible, namely
\begin{theorem}
\label{theo:main_network}
Given a network steerable assemblage $\tilde{\sigma}_B^{\vec{a}|\vec{x}}$, there exists
a measurement with a POVM element, 
without loss of generality,
$M^{b=0}_{BB'}$, such that the corresponding assemblage $\tilde{m}_{B'}^{b=0,\vec{a}|\vec{x}}$ is steerable in the same network.
\end{theorem}
Moreover (cf. below), it also follows that any witness certifying steering of the assemblage $\tilde{\sigma}^{\vec{a}|\vec{x}}_B$ can be used to detect the steerability of $\tilde{m}_{B'}^{b=0,\vec{a}|\vec{x}}$,
as any observable on the latter can be estimated using a basis of (possibly tomographically complete) inputs $\psi^y$ in Eq.~\eqref{eq:mdi_exp}. 
In summary, from the combination of Proposition~\ref{prop:meas_unsteer}
and Theorem~\ref{theo:main_network}, we conclude that
\begin{center}
{\em All network steerable assemblages can be certified measurement-device-independently.}
\end{center}

In order to prove Theorem~\ref{theo:main_network}, one can choose, in finite-dimensional systems, $M_{BB'}^{b=0}$ to be the projection onto the maximally entangled state, $\ket{\psi^+}=d^{-\frac{1}{2}} \sum_{i=0}^{d-1}\ket{ii}$, to obtain
\begin{align}
    M^{b=0}_{BB'}=\ketbra{\psi^+}_{BB'} \rightarrow \mathcal{M}^{b=0,\vec{a}|\vec{x}}_{B'}=\frac{1}{d}\left(\tilde{\sigma}^{\vec{a}|\vec{x}}_{B'}\right)^\intercal \;,
\end{align}
where $(\cdot)^\intercal$ is the transpose.
The theorem then follows from the fact that the associated assemblage is $\tilde{m}^{b=0,\vec{a}|\vec{x}}_{B'}\equiv\left(\tilde{\sigma}^{\vec{a}|\vec{x}}_{B'}\right)^\intercal$, together with the observation that an assemblage is steerable if and only if its transpose is steerable. In the continuous-variable case, a similar argument based on the projection onto two-mode squeezed states~\cite{abiuso2021measurement} can be used, which needs additional technical care. For further details on the complete CV proof of Theorem~\ref{theo:main_network}, we refer to the Appendix~\ref{app:steerproof}.




\section{Bipartite case: MDI Steering and semi-quantum games}
\subsection{Steerability of quantum states}



So far, we have been studying steering from assemblages, which are the local quantum statistics obtained by Bob.
Fundamentally, one necessary ingredient for an agent to steer another remote agent is their pre-shared quantum state.
It is consequently also natural to view steering as a property of quantum states (see, e.g., Refs.~\cite{Piani2015PRL,Hsieh2016PRA}) and to ask whether we can characterize steering at this level.
To be more precise, for simplicity, let us now consider a bipartite setting $AB$.
Then, we say a state $\rho_{AB}$ is {\em steerable} if it can induce a steerable assemblage.
More formally, we have the following definition:\\

\begin{definition}[Steerable state]
\label{Def:steerable state}
{\em A state $\rho_{AB}$ is {\em steerable} (from $A$ to $B$) if there exist a set of POVMs ${\bf E}\coloneqq\{E^{a|x}_A\}_{a,x}$ on $A$ (namely, for each $x$, the set $\{E^{a|x}_A\}_a$ is a POVM) such that the assemblage defined by 
\begin{align}\label{Eq: state induced assemblage}
\tilde{\sigma}_B^{a|x}(\rho_{AB},{\bf E})\coloneqq{\Tr_A}\left[(E_A^{a|x}\otimes\mathbb{I}_B)\rho_{AB}\right]
\end{align}
is steerable.
$\rho_{AB}$ is {\em unsteerable} if it is not steerable.}
\end{definition}

Clearly, a state $\rho_{AB}$ can induce infinitely many assemblages on Bob's side.
Operationally, in a steering experiment, all possible assemblages induced by pre-sharing $\rho_{AB}$ can be constructed by using {\em local operations and shared randomness} (LOSR) to map $\rho_{AB}$ into an assemblage~\cite{Schmid2020Quantum, Cavalcanti2013PRA}.
Mathematically, they are collections of mappings $\{\mathcal{E}^{a|x}\}_{a,x}$ of the following form~\cite{Cavalcanti2013PRA}: 
\begin{align}
    \rho_{AB}&\mapsto\mathcal{E}^{a|x}({\rho}_{AB})\nonumber\\
    &\coloneqq \int_\Lambda d\mu(\lambda) \; \mathcal{F}^\lambda_B\left[ \tilde{\sigma}_B^{a|x}(\rho_{AB},{\bf E}^\lambda)\right]\nonumber\\
    &=\int_\Lambda d\mu(\lambda) \; \mathcal{F}^\lambda_B\circ\mathrm{Tr}_A\left[ (E^{a|x,\lambda}_A\otimes\mathbb{I}_B){\rho}_{AB} \right],
\end{align}
where $\mathcal{F}^{\lambda}_B$'s are channels (namely, complete-positive trace-preserving linear maps) acting on $B$, ${\bf E}^\lambda\coloneqq\{E^{a|x,\lambda}_A\}_{a,x}$ is a set of POVMs for every $\lambda$, and $d\mu(\lambda)$ is a probability measure with respect to the shared random variable $\lambda\in\Lambda$. 
Physically, in a bipartite steering experiment, Alice and Bob can construct the assemblage $\mathcal{E}^{a|x}(\rho_{AB})$ by pre-sharing the randomness $\lambda\in\Lambda$ and then using this hidden variable $\lambda$ to run different sub-experiment with local measurement 
${\bf E}^\lambda$
in $A$ and local channel $\mathcal{F}_B^\lambda$ in $B$.

For a given state ${\rho}_{AB}$, we denote all possible assemblages induced by it on Bob's side collectively as the set
\begin{align}\label{Eq:LOSR induced state assemblages}
    \mathcal{S}_B({\rho}_{AB}) \coloneqq \left\{\mathcal{E}^{a|x}({\rho}_{AB})\;\Big|\;\{\mathcal{E}^{a|x}\}_{a,x}:{\rm LOSR} \right\}.
\end{align}
This naturally introduces a partial order over states: for two states $\rho_{AB}$ and $\phi_{AB}$, we write 
\begin{align}
\rho_{AB} \succeq_{\mathrm{st} } \phi_{AB}
\end{align}
if and only if, for every 
assemblage $\tilde{\phi}^{a|x} \in \mathcal{S}_B({\phi}_{AB})$,
\begin{align}\label{Eq:norm for conversion}
   \inf_{\tilde{\rho}^{a|x} \in \mathcal{S}_B({\rho}_{AB})}
   \sup_{x\in \mathcal{X}}
   \int_{\mathcal{A}}d{a}\;\left\lVert \tilde{\rho}^{a|x} - 
   \tilde{\phi}^{a|x}  
   \right\rVert_1 =0,
\end{align}
where $\norm{\cdot}_1$ is the trace norm.
Physically, this means that $\rho_{AB}$ outperforms (or equals) $\phi_{AB}$ in any steering experiment, as every assemblage induced by $\phi_{AB}$ can also be produced by $\rho_{AB}$ via LOSR.
It is then natural to ask whether unsteerable states are the least powerful.
This is indeed the case: by setting $\mathcal{F}^\lambda_B(\cdot) = \sigma^\lambda_B \Tr(\cdot)$ as a measure-and-prepare channel of the state $\sigma^\lambda_B$, we have, for every $\rho_{AB}$,
\begin{align}
     \mathcal{E}^{a|x}({\rho}_{AB})
     =&  \int_\Lambda d\mu(\lambda) \; \mathrm{Tr}\left[ \left({E}^{a|x,\lambda}_A\otimes \mathbb{I}_B\right) {\rho}_{AB} \right]\sigma^\lambda_B \notag\\
     \eqqcolon&\int_\Lambda d\mu(\lambda) \; P(a|x,\lambda)\sigma^\lambda_B,
\end{align}
where we have defined $P(a|x,\lambda)\coloneqq\mathrm{Tr}\left[ \left({E}^{a|x,\lambda}_A\otimes \mathbb{I}_B\right) {\rho}_{AB} \right]$. This is
in the bipartite local-hidden-state (LHS) model; that is, NLHS with only two parties.
Hence, LOSR can generate every LHS model from any pre-shared state. Together with the fact that $\mathcal{E}^{a|x}(\phi_{AB}^{\rm LHS})$ is unsteerable for every LOSR $\{\mathcal{E}^{a|x}\}_{a,x}$ and unsteerable state $\phi_{AB}^{\rm LHS}$, we conclude that, for every state $\rho_{AB}$,
\begin{align}
\mathcal{S}_B({\phi}_{AB}^{\rm LHS})\subseteq \mathcal{S}_B({\rho}_{AB})\quad\forall\;{\phi}_{AB}^{\rm LHS}:{\rm unsteerable}. 
\end{align}
That is, 
\begin{align}
\text{$\rho_{AB}\succeq_{\mathrm{st}}\phi_{AB}^{\rm LHS}$ for all unsteerable states ${\phi}_{AB}^{\rm LHS}$.} 
\end{align}

\subsection{Comparing steerable states measurement-device-independently}
Now, it is of foundational value to ask whether the MDI approach developed in this work can also describe this partial order.
In other words, we ask:
\begin{center}
{\em Is there an MDI characterization of the partial order ``$\succeq_{\mathrm{st}}$''?}
\end{center}
A successful answer can allow us to go beyond certification by measurement-device-independently {\em comparing} steerability of different states and hence also the assemblages induced by them.
In this section, we answer this question in the positive by generalizing the approaches of Refs.~\cite{cavalcanti2013entanglement, buscemi2012all} from finite dimension to the CV regime.
Namely, not just assemblages, the MDI approach can also fully capture steering even at the level of states.
To this end, we consider a scenario termed {\em steering game} as detailed as follows.
Without loss of generality, here we take Alice's side ($A$) as the untrusted side (as given in Definition~\ref{Def:steerable state}).

\begin{definition}[Steering game]
\em
In a bipartite setting with two agents $A$ (Alice) and $B$ (Bob), a {\em steering game}, denoted by $\mathcal{G}$, is a collection containing the following ingredients:
\begin{enumerate}
    \item Four index sets $\mathcal{A} = \{a\}, \mathcal{B}= \{b\}, \mathcal{X} = \{x\},\mathcal{Y} = \{y\}$. Each of them must be $\mathbb{R}$, $\mathbb{N}$, or a finite Cartesian product between them.
    \item A probability distribution $P_\mathcal{X}(x)$ on $\mathcal{X}$ and a state ensemble $\{{\psi}^y_{B'}, P_\mathcal{Y}(y)\}_y$ on $\mathcal{Y}$ (namely, each ${\psi}^y_{B'}$ is a state, and $P_\mathcal{Y}(y)$ is a probability distribution on $\mathcal{Y}$).
    \item A bounded pay-off function $\mathcal{P}: \mathcal{A}\times \mathcal{B}\times \mathcal{X}\times \mathcal{Y} \rightarrow \mathbb{R}$. 
\end{enumerate}
A given steering game $\mathcal{G}$ aims to test the steerability of a given state $\rho_{AB}$ shared by Alice and Bob, which works as follows.
In each round, the referee picks indices $x \in\mathcal{X}$ and $y\in\mathcal{Y}$ according to the probability distributions $P_\mathcal{X}(x)$ and $P_\mathcal{Y}(y)$.
The referee sends the index $x$ to Alice and the quantum state $\psi^y_{B'}$ to Bob.
Alice uses $x$ as her input as the standard steering experiment.
She then measure her local party by pre-fixed POVMs ${\bf E}=\{E_A^{a|x}\}_{a,x}$, producing the outcome $a$ and hence the state assemblage $\tilde{\sigma}_B^{a|x}(\rho_{AB},{\bf E})$ on Bob's side [as given in Eq.~\eqref{Eq: state induced assemblage}]. 
Then, Bob performs a pre-fixed bipartite POVM ${M}^{b}_{BB'}$ jointly on $\tilde{\sigma}_B^{a|x}(\rho_{AB},{\bf E})\otimes\psi_{B'}^y$ and obtain the outcomes $b$ [resulting in statistics in the form of Eq.~\eqref{eq:mdi_exp}].
\end{definition}

As a result, their performance in the steering game $\mathcal{G}$ is evaluated through the pay-off function $\mathcal{P}(a,b,x,y)$.
The average score of the player with the state $\rho_{AB}$ is determined by 
\begin{align}
    \mathbb{E}_{\mu}[ \mathcal{P}] = \int_{\mathcal{A},\mathcal{B},\mathcal{X},\mathcal{Y}} d\mu(a,b,x,y) \;\mathcal{P}(a,b,x,y)
\end{align}
where $d\mu(a,b,x,y) \coloneqq \mu(a,b,x,y)da\;db\;dx\;dy$ and
\begin{align}
    &\mu(a,b,x,y) \coloneqq\nonumber\\
    &\quad\quad P_\mathcal{X}(x)P_\mathcal{Y}(y)\Tr\left[{M}_{BB'}^b\left(\tilde{\sigma}_B^{a|x}(\rho_{AB},{\bf E})\otimes\phi_{B'}^y\right)\right].
\end{align}
As a side remark, the integral is taken with respect to the Lebesgue measure, and we replace it by the discrete sum when the indices are countable.

For a given $\rho_{AB}$, let $\mathbb{P}(\rho,\mathcal{G})$ be the set of all possible convex mixtures of probability distributions $\mu$'s induced by $\rho_{AB}$ in $\mathcal{G}$, which can be explicitly written as
\begin{align}\label{Eq: P set definition}
    \mathbb{P}(\rho_{AB},\mathcal{G}) \coloneqq \mathrm{conv}\Big\{&\mu(a,b,x,y)={P_\mathcal{X}(x)P_\mathcal{Y}(y)}\notag\\
    &\times \mathrm{Tr}\left[\left(E_{A}^{a|x}\otimes {M}_{BB'}^{b}\right) \left( \rho_{AB} \otimes \psi_{B'}^y \right)\right] \notag\\
     &\Big|\;\{E_{A}^{a|x}\}_{a,x}, \{{M}_{BB'}^{b} \}_b:\mathrm{POVMs}
    \Big\}.
\end{align}
Then, the maximum score of $\rho_{AB}$ in the steering game $\mathcal{G}$ is given by
\begin{align}
    \wp(\rho_{AB},\mathcal{G}) \coloneqq \sup_{\mu \in \mathbb{P}(\rho,\mathcal{G})}\mathbb{E}_{\mu}[ \mathcal{P}]. 
\end{align}
From here, we define a partial order by comparing the probability distributions $\mu$'s that can be induced.
More formally, we write 
\begin{align}
\rho_{AB}\succeq_{\rm game} \phi_{AB}
\end{align}
if and only if
\begin{align}
\overline{\mathbb{P}}({\rho}_{AB}, \mathcal{G}) \supseteq \overline{\mathbb{P}}({\phi}_{AB}, \mathcal{G})\quad\forall\;\mathcal{G},
\end{align}
where $\overline{\mathbb{P}}$ denotes the closure of the set $\mathbb{P}$ in the $L^1$-norm topology.
Physically, this is the set of distributions that can be prepared, including those that can be asymptotically approximated with an arbitrary precision.
Additionally, if $\rho_{AB}$ can induce all distributions in $\overline{\mathbb{P}}({\phi}_{AB}, \mathcal{G})$ {\em without} asymptotic approximation, we further write 
\begin{align}
{\rho_{AB}\succ_{\rm game} \phi_{AB},}
\end{align}
which means
\begin{align}
\mathbb{P}({\rho}_{AB}, \mathcal{G}) \supseteq \overline{\mathbb{P}}({\phi}_{AB}, \mathcal{G})\quad\forall\;\mathcal{G}.
\end{align}
As one can already observe, these partial orders are both MDI, as they can be tested by only trusting the input states $\psi^y$.
Consequently, if one can use it to characterize the partial order ``$\succeq_{\rm st}$'' from {LOSR} transformations, then an MDI characterization can be established.
Furthermore, we note that the partial order $\rho_{AB}\succeq_{\rm game}\phi_{AB}$ by definition implies that $\rho_{AB}$ outperforms (or at least equals) $\phi_{AB}$ in every steering game.
In fact, we have the following:
\begin{lemma}\label{lemma:HB-thm}
For every states $\rho_{AB}$ and $\phi_{AB}$, we have that 
$\rho_{AB}\succeq_{\rm game}\phi_{AB}$ if and only if 
\begin{align}
\wp (\rho_{AB}, \mathcal{G}) \geq \wp( {\phi}_{AB} , \mathcal{G}) \quad\forall\;\mathcal{G}.
\end{align}
\end{lemma}
See Appendix~\ref{App:stering game proofs} for the detailed proof.
Hence, the partial order ``$\succeq_{\rm game}$'' is equivalent to comparing the performance in all steering games.

Now, we present our main result in this section, which are necessary/sufficient conditions for LOSR transformation (whose proofs are located in Appendix~\ref{App:stering game proofs}):
\begin{theorem}\label{Result:steering game combined}
Consider two states $\rho_{AB}$ and $\phi_{AB}$.
Then, we have
\begin{align}\label{theo:steering_monotone}
{\rho}_{AB} \succeq_{\mathrm{st} }{\phi}_{AB}\quad&\Longrightarrow\quad{\rho}_{AB} \succeq_{\rm game} {\phi}_{AB},\\
\label{theo:assemblage_conversion}
{\rho}_{AB} \succ_{\rm game} {\phi}_{AB}\quad&\Longrightarrow\quad{\rho}_{AB} \succeq_{\mathrm{st} }{\phi}_{AB}.
\end{align}
\end{theorem}
Hence, the partial order ``$\succeq_{\rm game}$'' and its strong version ``$\succ_{\rm game}$'' jointly characterize
the partial order ``$\succeq_{\rm st}$,'' which is from the {LOSR} transformations.
Physically, this means that comparing the assemblages that can be prepared via {LOSR} (namely, ``$\succeq_{\rm st}$'') can be fully determined in an MDI way (via ``$\succeq_{\rm game}$'' and ``$\succ_{\rm game}$''), which thus answers the question posted at the beginning of this section.

Despite this result has successfully uncovered an MDI characterization of the partial order ``$\succeq_{\rm st}$,'' there is still an open technical problem.
Namely, it is still unknown whether one can prove that the partial order ``$\succeq_{\rm st}$'' is equivalent to the partial order ``$\succeq_{\rm game}$.''
We leave this open problem for future projects due to its mathematical nature, which is beyond the scope of this work, as we focus on the physical implications here.




\section{Practical MDI protocols for continuous-variable network steering}
\label{sec:MDI_GAUSS}

We start this section considering standard bipartite steering, i.e., $n=2$, with one untrusted party $A$ (Alice), and one semi-trusted $B$ (Bob). 
In the CV bosonic regime, we focus on the case of CV Gaussian quantum states and their steering properties.
Two main criteria are known to hold for the bipartite steering of Gaussian states (under Gaussian measurements), namely, the {\em Reid's criterion}~\cite{reid1989demonstration} and the {\em covariance matrix criterion}~\cite{wiseman2007steering} (which are equivalent~\cite{wiseman2007steering}; see also Appendix~\ref{app:CV_steer} for details). 

\begin{tcolorbox}[title=Notation and choice of units for CV states.]
In this work, we use standard units for bosonic operators and quadratures, in which
\begin{align}\label{Eq: x p def}
    \hat{x}=\frac{\hat{a}+\hat{a}^\dagger}{\sqrt{2}}\;,\quad 
    \hat{p}= \frac{\hat{a}-\hat{a}^\dagger}{\sqrt{2}i}\;,
\end{align}
where $a$, $a^\dagger$ are the canonical annihilation and creation operators.
It follows that $[\hat{x},\hat{p}]=i$ (where we have set $\hbar=1$), and the uncertainty principle reads ${\rm Var}[\hat{p}]{\rm Var}[\hat{x}]\geq \frac{1}{4}$. 
In these units, the quadratures' variance on coherent states, including the vacuum, is equal to $\frac{1}{2}$.
Coherent states are defined as
\begin{align}\label{Eq:coherent state definition}
    \ket{\alpha}=e^{i \sqrt{2}\alpha_p \hat{x}-i\sqrt{2}\alpha_x\hat{p}}\ket{0}=e^{\alpha \hat{a}^{\dagger}-\alpha^*\hat{a}}\ket{0}\;,
\end{align}
where $\alpha=\alpha_x+ i\alpha_p$. A coherent state satisfies $\langle \hat{x}_\alpha\rangle=\sqrt{2}\alpha_x$ and $\langle \hat{p}_\alpha\rangle=\sqrt{2}\alpha_p$.
\end{tcolorbox}

Reid's criterion considers a two-input scenario for the untrusted party, Alice, trying to estimate Bob's quadratures, conditioned on her outputs. For clarity, we label Alice's measurement choices as $x$ and $p$ (for measuring two quadratures $\hat{x}$ and $\hat{p}$), for which she obtains as a result of the random variable $v_x$ or $v_p$. 
After knowing the value of her result, Alice can try to estimate the corresponding value of Bob's quadratures $\hat{x}_B$, or $\hat{p}_B$, guessing $x_{\rm est}(v_x)$ and $p_{\rm est}(v_p)$, respectively. We can therefore build Alice's error operators
\begin{align}
    \hat{x}_{\rm err}:=\hat{x}_B|_{v_x}-x_{\rm est}(v_x)\;, \\
    \hat{p}_{\rm err}:=\hat{p}_B|_{v_p}-p_{\rm est}(v_p)\;.
\end{align}
Reid's criterion says that unsteerable states satisfy the Heisenberg uncertainty principle for $\hat{x}_{\rm err},\hat{p}_{\rm err}$. That is, the inequality 
\begin{align}
    \langle \hat{x}_{\rm err}^2\rangle \langle \hat{p}_{\rm err}^2\rangle \geq \frac{1}{4}
    \label{eq:uncert_prod_main}
\end{align}
holds for unsteerable states regardless of Alice's choice of observables.
This immediately implies
\begin{align}
    \langle \hat{x}_{\rm err}^2\rangle + \langle \hat{p}_{\rm err}^2\rangle \geq 1\;.
    \label{eq:uncert_sum_main}
\end{align}
Notice that $\langle\hat{x}_{\rm err}^2\rangle$ is minimized when $x_{\rm est}(v_x)=\langle\hat{x}_B\rangle|_{v_x}$, and similarly for $\hat{p}_{\rm err}$. That is, the optimal estimation by Alice is that of guessing the average value of Bob's quadrature conditioned on Alice's outcome. In such a case, one has
$\langle\hat{x}_{\rm err}^2\rangle \equiv {\rm Var}[\hat{x}_B]|_{v_x}$, 
$\langle\hat{p}_{\rm err}^2\rangle \equiv {\rm Var}[\hat{p}_B]|_{v_p}$.
\\

\begin{figure}
    \centering
    \includegraphics[width=0.5\textwidth]{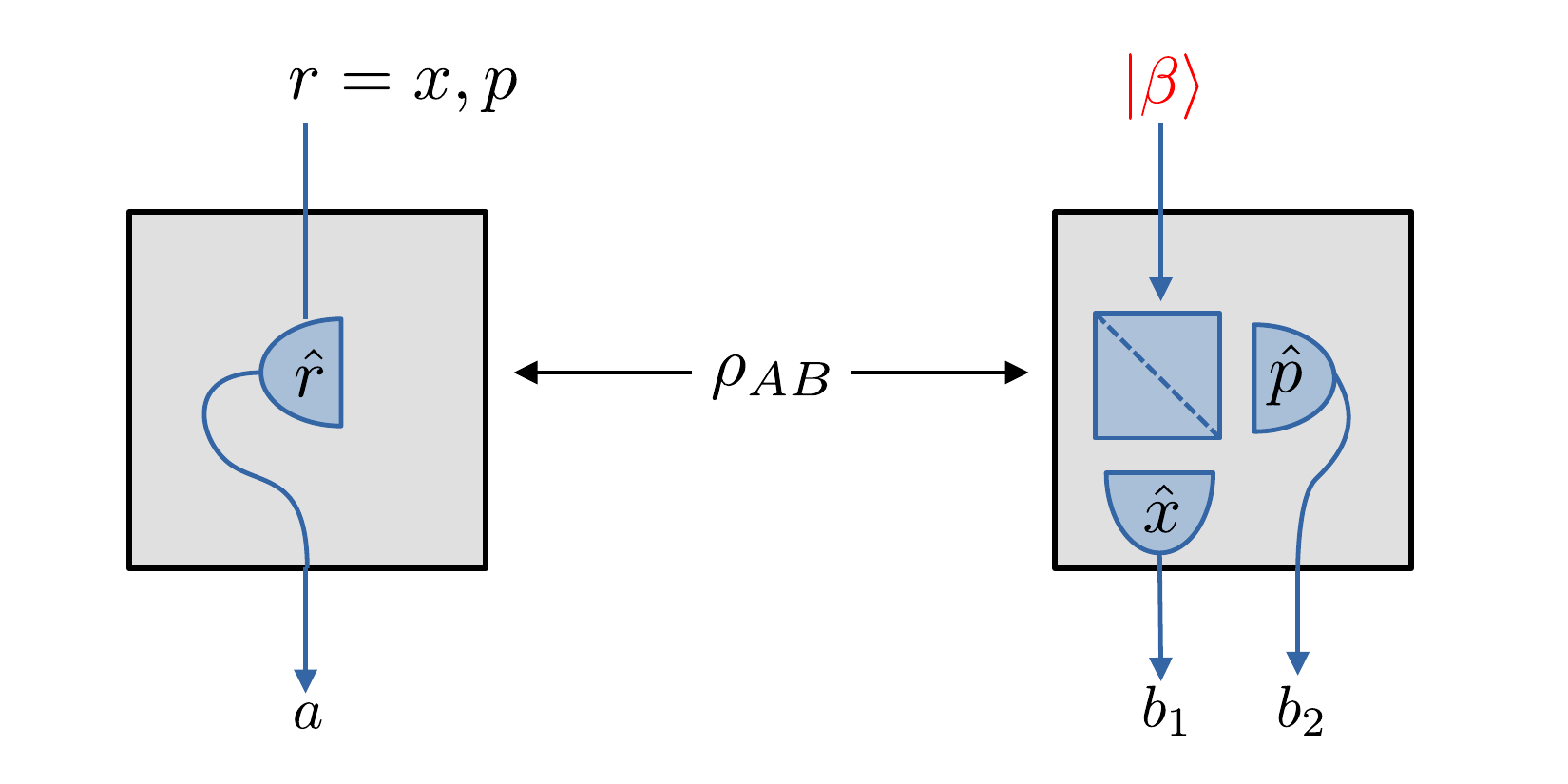}
    \caption{Setup for a practical CV MDI steering witnessing experiment based on Gaussian quantum optics. The protocol is ``almost device-independent'': only the input coherent state $\ket{\beta}$ is trusted (which is highlighted in red). The corresponding witness, as given in Eq.~\eqref{eq:EB_bound_main}, can be violated by any Gaussian steerable state via simple homodyne measurements (in the Figure, measurements of quadratures $\hat{x}$ or $\hat{p}$, and their corresponding outputs.}
    \label{fig:setup}
\end{figure}

\subsection{MDI witnesses for bipartite CV steering}

Starting from Reid's criterion, in the following, we show that any (bipartite) Gaussian steerable state can be detected in an MDI-steering scenario. 
That is, a scenario with one party (Alice) acting on a completely untrusted black-box, and Bob only trusts coherent states inputs to the experiment. The protocol we propose for detecting steering only uses Gaussian states and measurements, while unsteerable states cannot be used to violate the witness regardless.
Consider the scenario represented in Fig.~\ref{fig:setup}. Alice and Bob share a bipartite CV state. Alice has a dichotomic choice of measurements (labelled by $r=\{x,p\}$ for convenience) and outputs $a$. Bob inputs in their measurement device a random coherent state $\ket{\beta}$, which is the only trusted part of the experiment, and outputs two values $b_1,b_2$.
In this scenario, we propose the following MDI-steering witness for CV state: if $\rho_{AB}$ is unsteerable, one has the quantity
\begin{align}
    \mathcal{W}\coloneqq\left\langle \left( b_1 - \frac{a}{\sqrt{2}}-\beta_x\right)^2\right\rangle\bigg|_{r=x} + 
    \left\langle\left( b_2 + \frac{a}{\sqrt{2}}-\beta_p\right)^2\right\rangle \bigg|_{r=p} 
    \label{eq:witness}
\end{align}
to be lower bounded by $\mathcal{W}\gtrsim 1$, where the exact lowerbound can be made arbitrarily close to $1$, depending on the input distribution of $\beta$. 
More precisely,
\begin{theorem}
\label{res1}
    In the bipartite MDI scenario of Fig.~\ref{fig:setup}, assuming $\rho_{AB}$ is unsteerable and $\beta$ is sampled randomly from a distribution with width $\Delta_\beta$, the minimum value of $\mathcal{W}$ as defined in Eq.~\eqref{eq:witness} is bounded by
\begin{align}
\label{eq:EB_bound_main}
    \mathcal{W} \geq \frac{\Delta_\beta^2}{1+\Delta_\beta^2}\;.
\end{align}
The sample distribution of $\beta$ is assumed to be Gaussian for simplicity, i.e., $P(\beta)=(\pi\Delta^2_\beta)^{-1}{\rm Exp}[-|\beta|^2/\Delta^2_\beta]$, and equivalent bounds can be proven for other forms of $P(\beta)$~\cite{abiuso2021measurement,abiuso2023verification}.
\end{theorem}
We provide a full proof of Theorem~\ref{res1} in Appendix~\ref{app:MDI_gauss_wit}. The gist of the proof is as follows: the combinations $b_1-\frac{a}{\sqrt{2}}$ and $b_2+\frac{a}{\sqrt{2}}$ can be regarded as effective estimators of $\beta_r$ for $r=x$ and $r=p$, respectively. The unsteerability of $\rho_{AB}$, via its corresponding LHS model, makes the effective estimation strategy equivalent to a remote joint estimation of both $\beta_x$ and $\beta_p$ made by Alice.
One can then bound the minimum joint error in the estimation of unknown coherent states quadratures using multiparameter Bayesian Cram\'er-Rao bounds~\cite{yuen1973multiple}, which take into account the single-shot incompatibility of $\hat{x}$ and $\hat{p}$ measurements (cf. Refs.~\cite{genoni2013optimal,morelli2021bayesian,abiuso2023verification}).

Having verified that inequality in Eq.~\eqref{eq:EB_bound_main} cannot be violated without the use of bipartite CV steering, we also show a dual statement, namely
\begin{proposition}
\label{prop:gauss_viola}
Any Gaussian steerable $\rho_{AB}$ is sufficient to violate
the inequality in Eq.~\eqref{eq:EB_bound_main}
via Gaussian measurements. The protocol, as pictured in Fig.~\ref{fig:setup}, only uses beam splitters and homodyne measurements.
\end{proposition}

In the protocol, Bob's station mixes the input coherent state with his share of the source in a 50:50 beam splitter. the outputs of Bob are therefore represented by the observables
\begin{align}
    \hat{b}_1 =\frac{\hat{x}_B+\hat{x}_\beta}{\sqrt{2}}\;,\quad
    \hat{b}_2 =\frac{-\hat{p}_B+\hat{p}_\beta}{\sqrt{2}}\;.
\end{align}
At the same time, Alice measures the two possible quadratures $\hat{r}_A=\hat{x}_A,\hat{p}_A$, resulting in  $\mu_r$, and outputs
\begin{align}
    a= & r_{\rm est}(\mu_r)\;,
\end{align}
that is, $ a= \langle\hat{x}_B\rangle|_{v_x} $ if $r=x$, while $a= \langle\hat{p}_B\rangle|_{v_p} $ if $r=p$.
When Alice and Bob follow this protocol, the $x$-term in the witness defined in Eq.~\eqref{eq:witness}, namely the average of  \mbox{$( b_1 - \frac{a}{\sqrt{2}}-\beta_x)^2$}, becomes
\begin{align}
   \left\langle \left( \frac{\hat{x}_B+\hat{x}_\beta}{\sqrt{2}}- \frac{\langle\hat{x}_B\rangle|_{v_x}}{\sqrt{2}}-\beta_x\right)^2\right\rangle 
   =\frac{\langle\hat{x}^2_{\rm err}\rangle}{2} + \frac{{\rm Var} [\hat{x}_\beta]}{2}\;,
\end{align}
and similarly the $p$-term is equal to $\frac{\langle\hat{p}^2_{\rm err}\rangle}{2} + \frac{{\rm Var} [\hat{p}_\beta]}{2}$. In our units, the variance  of the quadratures on coherent states is equal to ${\rm Var}[\hat{x}_\beta]={\rm Var}[\hat{p}_\beta]=\frac{1}{2}$. From this, {with such a protocol,} it follows that the entire witness becomes
\begin{align}
    \mathcal{W}_{\rm protocol}=\frac{1}{2}+{\frac{\langle \hat{x}_{\rm err}^2\rangle + \langle \hat{p}_{\rm err}^2\rangle}{2}.}
\end{align}
It then further follows that, if $\rho_{AB}$ is steerable and such that $\langle \hat{x}_{\rm err}^2\rangle + \langle \hat{p}_{\rm err}^2\rangle<1$, we have
\begin{align}
    \mathcal{W}_{\rm protocol}< 1\;,
\end{align}
which violates the inequality in Eq.~\eqref{eq:EB_bound_main} when $\Delta_\beta$ is sufficiently large. For a pure two-mode squeezed state, which can be expressed in the Fock basis as
\begin{align}
\label{eq:TMSS}
   \ket{t^{(\gamma)}}\coloneqq\frac{1}{\cosh(\gamma)}\sum_{j=0}^{\infty} \tanh{(\gamma)}^j
\ket{jj}\;,
\end{align} we have $\mathcal{W}_{\rm protocol}=\frac{1}{2}+\frac{1}{2\cosh{2\gamma}}$  ({cf.~Appendix}~\ref{sec:TMSS}), and as such any amount of squeezing is sufficient to violate the MDI steering witness.
Finally, in order to prove Proposition~\ref{prop:gauss_viola}, we notice that any pair of quadratures violating the uncertainty relation in Eq.~\eqref{eq:uncert_prod_main} can be mapped  via  a local squeezing of the form $\hat{x}_{\rm err}\rightarrow \hat{x}'_{\rm err}=\kappa\hat{x}_{\rm err}$, $\hat{p}_{\rm err}\rightarrow \hat{p}'_{\rm err}=\hat{p}_{\rm err}/\kappa$ violating Eq.~\eqref{eq:uncert_sum_main}, for some real value of $\kappa$. Moreover, the (Gaussian and symplectic) $\kappa$-squeezing can be performed either on the $B$-mode directly, or by rescaling of the beam splitter transimissivity and input coherent distribution (details in Appendix~\ref{app:protocol_viola}).

\subsection{A simple line-network example}
The MDI steering protocol presented above serves as a basis for the example we present in the following: we consider a simple $n$-partite network formed by $n-1$ untrusted parties $A_i:i=1,\dots n-1$ and a party $B$ in a line, where $A_1$ and $B$ corresponding to the extremal points of the line, each edge of which is a bipartite source shared by $A_i$ and $A_{i+1}$ (including $A_{n-1}$ and $B$). As in the bipartite case, $A_{1}$ has an input $r=\{x,p\}$ to the experiment and a single output $a^{(1)}$, whereas each $A_{i>1}$ has no inputs and two outputs $a^{(i)}_{1,2}$. $B$ injects trusted coherent states $\ket{\beta}$ and has two outputs $b_{1,2}$. In such a scenario, we define the MDI witness 
\begin{widetext}
\begin{align}
    \mathcal{W}_{\rm line}:=\left\langle \left( b_1 +a^{(n-1)}_1+\dots+a^{(2)}_1- \frac{a^{(1)}}{\sqrt{2}}-\beta_x\right)^2\right\rangle\bigg|_{r=x} + 
    \left\langle\left( b_2 +a^{(n-1)}_2+\dots+a^{(2)}_2+ \frac{a^{(1)}}{\sqrt{2}}-\beta_p\right)^2\right\rangle \bigg|_{r=p} \;,
    \label{eq:witness_line}
\end{align}
\newpage
\end{widetext}
which corresponds to a line-network generalization of $\mathcal{W}$ defined in Eq.~\eqref{eq:witness} above.
Then, we have the following result:
\begin{proposition}
    Under the same assumptions of Theorem~\ref{res1}, if the underlying assemblage $\tilde{\sigma}_B^{\vec{a}|\vec{x}}$ is line-network unsteerable, the MDI witness $\mathcal{W}_{\rm line}$ satisfies the same bound
    \begin{align}
    \mathcal{W}_{\rm line} \geq \frac{\Delta_\beta^2}{1+\Delta_\beta^2}\;.
    \label{eq:line_ineq}
\end{align}
Moreover, a line network consisting of squeezed two-mode states and balanced heterodyne detectors can be used to violate the inequality in Eq.~\eqref{eq:line_ineq}.
\end{proposition}

The proof of Eq.~\eqref{eq:line_ineq} immediately follows by considering the combinations $a^{(n-1)}_{1,2}+\dots+a^{(2)}_{1,2}\pm \frac{a^{(1)}}{\sqrt{2}}=:a$ as a single output of the super-party containing all $A_i$'s, and the observation that any assemblage $\sigma^{\vec{a}|\vec{x}}_B$ presenting a NLHS model, trivially has a bipartite LHS model; as such, $\sigma_B^{\vec{a}|\vec{x}}$ is unsteerable in the bipartition $\vec{A}|B$. 

The same inequality in Eq.~\eqref{eq:line_ineq} can instead be violated using Gaussian states and measurements. To show this, we can consider an initially shared quantum state consisting in a chain of bipartite two-mode squeezed states [as given in Eq.~\eqref{eq:TMSS}] shared by all adjacent parties 
\begin{align}
\label{eq:line_source}
    \rho_{A_1,\dots,A_{n-1},B}=t^{(\gamma_1)}_{A_1A'_2}\otimes t^{(\gamma_2)}_{A''_2A'_3}\otimes \dots\otimes t^{(\gamma_{n-1})}_{A''_{n-1}B}\;.
\end{align}
By then choosing the homodyne measurements preceded by a 50:50 beam splitter at each station (except $A_1$ who only has a single incoming mode) the parties can measure
\begin{align}\nonumber
    \hat{A}_1(r)&=\hat{r}_{A_1} \quad (r=x,-p)\;,\\ \nonumber
    \hat{A}_{i>1}&=\left\{ \frac{\hat{x}_{A'_i}-\hat{x}_{A''_i}}{\sqrt{2}}, \frac{-\hat{p}_{A'_i}-\hat{p}_{A''_i}}{\sqrt{2}} \right\}\;,\\
    \hat{B}&=\left\{ \frac{\hat{x}_{B}+\hat{x}_{\beta}}{\sqrt{2}}, \frac{-\hat{p}_{B}+\hat{p}_{\beta}}{\sqrt{2}}\right\}\;.
    \label{eq:line_measurements}
\end{align}
We evaluate the witness in Eq.~\eqref{eq:witness_line} with this choice of measurements. By splitting 
${\mathcal{W}}_{\rm line}=\langle\hat{\mathcal{W}}_{{\rm line},x}\rangle+\langle\hat{\mathcal{W}}_{{\rm line},p}\rangle$ for rounds in which $r=x$ and $r=p$ respectively, we find, by inserting the measurement choice in Eq.~\eqref{eq:line_measurements},
\begin{align}
    &\hat{\mathcal{W}}_{{\rm line},x}= \nonumber\\
    &\qquad\left(\frac{1}{\sqrt{2}}\bigg( \hat{x}_{B}+\hat{x}_{\beta} + \big(\sum_{i=2}^{n-1}\hat{x}_{A'_i}+\hat{x}_{A''_i}\big)-\hat{x}_{A_1} \bigg)-\beta_x\right)^2\\
    &\hat{\mathcal{W}}_{{\rm line},p}= \nonumber\\
    &\qquad\left(\frac{1}{\sqrt{2}}\bigg( -\hat{p}_{B}+\hat{p}_{\beta} + \big(\sum_{i=2}^{n-1} -\hat{p}_{A'_i}-\hat{p}_{A''_i}\big)-\hat{p}_{A_1} \bigg)-\beta_p\right)^2
\end{align}
By then using the tensor product of two-mode squeeze states as in Eq.~\eqref{eq:line_source}, we find
\begin{align}
    \langle\hat{\mathcal{W}}_{{\rm line},x}\rangle|_{\rho}=\langle\hat{\mathcal{W}}_{{\rm line},p}\rangle|_{\rho}=\frac{1}{4}+\frac{1}{2}\sum_{i=1}^{n-1}e^{-2\gamma_i}\;,
\end{align}
which immediately yields a witness violation
$\mathcal{W}_{\rm line} < 1$ whenever the states are sufficiently squeezed to satisfy $\sum_{i=1}^{n-1}e^{-2\gamma_i}<1$.

\section{Conclusion and Discussion}
In this work, we presented a comprehensive study of quantum steering and its measurement-device-independent (MDI) detection, in general networks and for both finite- and infinite-dimensional systems. 
We considered a generalized notion of steering in network scenarios, in which all-except-one parties are untrusted and treated as black-boxes, while a single party $B$ performs device-dependent analysis. The MDI lifting then consists in partially eliminating the trust in $B$'s devices, in particular the choice of measurement operators, while only allowing the calibrated preparation of fiduciary input states. This thus represents {\em the minimal} semi-device-independent scenario which is not entirely device-independent.

As the first main result, we showed that {\em all} network steering can be witnessed in its corresponding MDI scenario. This result holds regardless of the networks' topology, and 
(while being slightly more technical compared to the finite-dimensional case) also holds for continuous-variable quantum systems.

In the bipartite case, we extended our analysis beyond detection and showed that 
the processes of operationally preparing assemblages with a given bipartite state (which is by definition device-dependent) can  be characterized in 
an MDI framework. More precisely,
for two states, their abilities to induce assemblages can be compared by an MDI method.
This result shows that our MDI approach can not only detect steering from assemblages, but also distinguish which state is more resourceful in a steering experiment.
%

Finally, we showed that all Gaussian steerable states can be detected in an experimentally-friendly MDI-steering experiment
with a single trusted input coherent state and Gaussian measurements. Notably, even non-Gaussian steerable states violating the Reid's criterion [i.e., Eqs.~\eqref{eq:uncert_prod_main} and~\eqref{eq:uncert_sum_main}] can be detected by the same inequality given in Eq.~\eqref{eq:witness}. We also provided an example of the steering detection between the two ends of a quantum network in a line. We conclude that, while being inapplicable to fully DI protocols, Gaussian quantum optics becomes a viable option for quantumness certification in networks as soon as a single (Gaussian) input is trusted.

We believe this work puts quantum steering and its minimal-trust detection in a complete framework, paving the way to applications based on technologically-ready quantum optics,
and contributes to the undergoing understanding of the links between steering, MDI certification, and conditional metrology~\cite{yadin2021metrological,Larsen2026continuous}.

\acknowledgements
P.A.~acknowledges fundings from the Austrian Science Fund (FWF) projects I-6004 and ESP2889224.
C.-Y.H.~acknowledges support from the Leverhulme Trust Early Career Fellowship (``Quantum complementarity: a novel resource for quantum science and technologies'' with Grant No.~ECF-2024-310).
H.-Y.K.~is supported by National Science and
Technology Council (NSTC) with Grant No.~NSTC 112-
2112-M-003-020-MY3.

We declare that no AI/LLMs have been used to produce the results. All ideas, mathematical proofs, and figures are generated by the authors. We use AI (mainly Grammarly) to fix typos, check grammar, and polish sentences.

\hypertarget{equal}{\textsuperscript{$\star$}}~These two authors contributed equally.




\medskip
\bibliography{MDIsteer.bib}

\newpage
\onecolumngrid
\appendix

\section{Proof of Theorem~\ref{theo:main_network}}
\label{app:steerproof}

We start with a formal definition of network steering.

\begin{definition}
{\em{\bf(Network Steering)}
Consider, in a given causal network, the set of assemblages that can be decomposed as from Eq.~\eqref{eq:NLHS} (unsteerable assemblages), which is denoted as
\begin{align}
    \mathcal{U} \coloneqq \left\{\tilde\sigma_B^{\vec{a}|\vec{x}}\,\Bigg|\,\tilde\sigma_B^{\vec{a}|\vec{x}}= \sum_{\vec{\lambda}} \sigma^{(\vec{\lambda})}_B p_{\loc}(\vec{a}|\vec{x},\vec{\lambda})\mu(\vec{\lambda})\;\text{for some $\sigma^{(\vec{\lambda})}_B$, local probability $p_{\loc}(\vec{a}|\vec{x},\vec{\lambda})$, and probability measure $\mu(\vec{\lambda})$}\right\}.
\end{align}
Then, an assemblage $\tilde\sigma_B^{\vec{a}|\vec{x}}$ is said to be {\em network steerable} (or simply {\em steerable} in this work) if it does not pertain to the closure (in the trace norm topology) of $\mathcal{U}$; namely, $\tilde\sigma_B^{\vec{a}|\vec{x}}\notin\overline{\mathcal{U}}$}. Namely, there is an ``$\varepsilon$-ball'' centered at $\tilde\sigma_B^{\vec{a}|\vec{x}}$ that is completely outside $\overline{\mathcal{U}}$.
More precisely, what we mean here is that there exists some $\varepsilon_*>0$ such that
\begin{align}
\label{eqapp:out_clos}
    \left\|\tilde{\rho}_B^{\vec{a}'|\vec{x}'}-\tilde{\sigma}_B^{\vec{a}'|\vec{x}'}\right\|_1\leq \varepsilon_*\quad\text{for some $\vec{a}',\vec{x}'$} \implies \tilde{\rho}_B^{\vec{a}|\vec{x}}\notin \overline{\mathcal{U}}\;.
\end{align}
\end{definition}
We notice that in the absence of such condition on the closure of the set, one could in principle obtain a series of unsteerable assemblages ${\tilde{\rho}_\varepsilon^{\vec{a}|\vec{x}}}$ that approximate $\tilde{\sigma}^{\vec{a}|\vec{x}}$ to arbitrarily small experimental error.

We now restate the theorem for completeness:\\
{\bf Theorem~\ref{theo:main_network}.}
\emph{For a given network steerable $
\tilde{\sigma}_B^{\vec{a}|\vec{x}} $, there exists a measurement with a POVM element, 
without loss of generality,
$M^{b=0}_{BB'}$, such that the corresponding $\tilde{m}_{B'}^{b=0,\vec{a}|\vec{x}}$ as defined in Proposition~\ref{prop:meas_unsteer} is steerable in the same network.
}\\

In order to prove Theorem~\ref{theo:main_network} for general CV states, we will make use of the following lemma.
\begin{lemma}[Gentle measurement lemma \cite{khatri2024principlesquantumcommunicationtheory}]\label{gentle measurement lemma}
Let $\mathbb{I} \succeq M \succeq 0$ be an element of a discrete POVM. 
Then, it follows that
\begin{align}
    \frac{1}{2}\left\lVert \rho -  \frac{\sqrt{M}\rho \sqrt{M} }{\Tr[\rho M]} \right\rVert_{1} \leq \sqrt{1- \Tr[\rho M]} \;,
\end{align}
where $\norm{\cdot}_1$ is the trace norm.
\end{lemma}
\begin{proof} The proof can be found in Ref.~\cite{khatri2024principlesquantumcommunicationtheory}, and we repeat it here for completeness.
First, we have, from Fuchs-van de Graaf inequalities~\cite{Fuchs1999fidelity},
\begin{align}
     \frac{1}{2}\left\lVert \rho -  \frac{\sqrt{M}\rho \sqrt{M} }{\Tr[{\rho M}]} \right\rVert_1
     \leq \sqrt{1-F\left(  \rho, \frac{\sqrt{M}\rho \sqrt{M} }{\Tr[\rho M]}\right)  },
\end{align}
where $F$ is the fidelity defined as $F(\rho,\sigma)\coloneqq\left[\Tr\left(\sqrt{\sqrt{\rho}\sigma\sqrt{\rho}}\right)\right]^2$.
Then, by setting $\proj{\psi_{AB}}$ as any purification of $\rho$, we have
\begin{align}
    F\left(  \rho, \frac{\sqrt{M}\rho \sqrt{M} }{\Tr[\rho M]}\right) 
    &\geq F\left( \proj{\psi_{AB}},   \frac{(I\otimes \sqrt{M})\proj{\psi_{AB}} (I\otimes \sqrt{M}) }{\langle \psi_{AB}|I\otimes M | \psi_{AB} \rangle}\right) \notag\\
    &=  \frac{ |\langle \psi_{AB}|I\otimes \sqrt{M} | \psi_{AB}\rangle |^2 }{\langle \psi_{AB}| I\otimes M | \psi_{AB}\rangle } \notag\\
    &\geq \langle \psi_{AB}| I\otimes M |\psi_{AB}\rangle  \notag\\
    &= \Tr[\rho M] \;.
\end{align}
Here, we used the fact that $\sqrt{M}\succeq M$.
This thus concludes the proof.
\end{proof}\\

Notice that this lemma holds for both finite- and infinite-dimensional Hilbert spaces. 
This thus gives us the following corollary for continuous-variable bosonic states:

\begin{corollary}\label{gentle measurement lemma corollary}
Let $\rho$ be a quantum state. Then, for every $\varepsilon > 0$, there exists a parameter $\beta \in (0,\infty)$ such that
\begin{align}
  {\frac{1}{2}\left\lVert \rho -  \frac{e^{-\beta N}\rho e^{-\beta N} }{\Tr[\rho e^{-2\beta N}]} \right\rVert_1  < \varepsilon.}
\end{align}
Here we take $N=a^\dagger a$ as the bosonic number operator.
\end{corollary}

\begin{proof}
It is clear that $\mathbb{I} \succeq e^{-\beta N} \succeq 0$ for all $\beta >0$.
By the gentle measurement lemma (Lemma~\ref{gentle measurement lemma}), we have, 
for every $\beta \in [ 0,\infty)$,
\begin{align}\label{Eq:CV gentle measurement lemma proof 001}
    \frac{1}{2}\left\lVert \rho -  \frac{e^{-\beta N}\rho e^{-\beta N} }{\Tr[\rho e^{-2\beta N}]}\right\rVert_1 
    \leq \sqrt{ 1- \Tr[ \rho e^{-2\beta N}]}.
\end{align}
Now, for every $K\in\mathbb{N}$, we have
\begin{align}
1\ge\lim_{\beta\to0^+}\Tr[ \rho e^{-2\beta N}] = \lim_{\beta\to0^+}\sum_{n=0}^\infty\bra{n}\rho\ket{n}e^{-2\beta n}\ge \lim_{\beta\to0^+}\sum_{n=0}^K\bra{n}\rho\ket{n}e^{-2\beta n}=\sum_{n=0}^K\bra{n}\rho\ket{n}.
\end{align}
As this holds for every $K\in\mathbb{N}$, we conclude that 
\begin{align}
1\ge\lim_{\beta\to0^+}\Tr[ \rho e^{-2\beta N}]\ge\sum_{n=0}^\infty\bra{n}\rho\ket{n} = {\Tr}[\rho] = 1,
\end{align}
meaning that the upper bound in Eq.~\eqref{Eq:CV gentle measurement lemma proof 001} goes to zero when $\beta\to0^+$.
This thus implies the desired result.
\end{proof}\\

By using the gentle measurement lemma (Lemma~\ref{gentle measurement lemma}) and its corollary, we can now prove Theorem~\ref{theo:main_network} as follows.\\

\noindent{\bf Proof of Theorem~\ref{theo:main_network}.}
Assuming $\tilde{\sigma}_B^{\vec{a}|\vec{x}} \notin{\overline{\mathcal{U}}}$, where we take $\overline{\mathcal{U}}$ as $\mathcal{U}$'s closure in the trace norm topology. 
Then, as defined in Eq.~\eqref{eqapp:out_clos}, this implies the existence of
some $\varepsilon_*>0$ such that
\begin{align}
    \left\|\tilde{\rho}_B^{\vec{a}'|\vec{x}'}-\tilde{\sigma}_B^{\vec{a}'|\vec{x}'}\right\|_1\leq \varepsilon_*\quad\text{for some $\vec{a}',\vec{x}'$} \implies \tilde{\rho}_B^{\vec{a}|\vec{x}}\notin \overline{\mathcal{U}}\;.
\end{align}
Consider now a POVM containing the element
\begin{align}
    M_{BB'}^{b=0}|_\gamma=\ketbra{t^{(\gamma)}_{BB'}},
\end{align}
where recall from the main text that $\ket{t^{(\gamma)}_{BB'}}$ is a two-mode squeezed state [Eq.~\eqref{eq:TMSS}], expressed in the Fock basis as
\begin{align}
   \ket{t^{(\gamma)}_{BB'}} \coloneqq \frac{1}{\cosh(\gamma)}\sum_{j=0}^{\infty} \tanh{(\gamma)}^j
\ket{jj}_{BB'}\;.
\end{align}
One can then verify that, for every $\vec{a},\vec{x}$,
\begin{align}
\mathcal{M}_{B'}^{b=0,\vec{a}|\vec{x}}{\Big|_\gamma}={\Tr_B\left[\ketbra{t^{(\gamma)}_{BB'}}\left(\tilde{\sigma}^{\vec{a}|\vec{x}}_B\otimes\mathbb{I}_{B'}\right)\right]}={\frac{1}{\cosh{(\gamma)}^2}\tanh{(\gamma)}^{N}}{\left(\tilde{\sigma}_{B'}^{\vec{a}|\vec{x}}\right)^\intercal\tanh{(\gamma)}^{N},}
\end{align}
where note that $N$ is the number operator, and hence $\tanh{(\gamma)}^{N}$ is an operator.
Hence, we obtain (recall the definition of $\tilde{m}_{B'}^{b,\vec{a}|\vec{x}}$ from Proposition~\ref{prop:meas_unsteer})
\begin{align}
    \tilde{m}_{B'}^{b=0,\vec{a}|\vec{x}}{\Big|_\gamma}= \frac{\mathcal{M}_{B'}^{b=0,\vec{a}|\vec{x}}{\Big|_\gamma}}{\Tr[{M}_{BB'}^{b=0}{\Big|_\gamma}{(\sigma_B\otimes\mathbb{I}_{B'})}]} = \frac{\tanh{{(\gamma)}}^{N}\tilde{\sigma}_{B'}^{\vec{a}|\vec{x}\,\intercal}\tanh{{(\gamma)}}^{N}}{\Tr [\tanh{{(\gamma)}}^{N}{\left({\sigma}_{B'}\right)}^\intercal \tanh{{(\gamma)}}^{N}]}.
\end{align}
By identifying {$e^{-\beta}\equiv\tanh(\gamma)$,} we can use Lemma~\ref{gentle measurement lemma} and Corollary~\ref{gentle measurement lemma corollary} to conclude that, for the given $\varepsilon_*>0$ and any given pair of $\vec{a},\vec{x}$, there exists some $\gamma_*$ such that
\begin{align}
    {\left\|\tilde{m}_B^{b=0,\vec{a}|\vec{x}}{\Big|_{\gamma_*}}-\tilde{\sigma}_B^{\vec{a}|\vec{x}\,\intercal}\right\|_1}\leq {\varepsilon_*}\;,
\end{align}
which, combined with Eq.~\eqref{eqapp:out_clos} and the observation that steerability is transpose-invariant, proves Theorem~\ref{theo:main_network}.
$\blacksquare$

\section{Proofs Related to Steering Games}\label{App:stering game proofs}

We start with the proof of Lemma~\ref{lemma:HB-thm}, which is necessary for the proof of Theorem~\ref{Result:steering game combined}. For clarity, we remark that the distance that we used to determine conversion [i.e., Eq~\eqref{Eq:norm for conversion}] induces a topology different from the one considered in Eq.~\eqref{eqapp:out_clos}. They are two mathematical ways to measure the difference between assemblages, but will not change the physics.\\

\begin{proof}[Proof of Lemma~\ref{lemma:HB-thm}]
Note that $\rho_{AB}\succeq_{\rm game}\phi_{AB}$ implies 
$
\wp (\rho_{AB}, \mathcal{G}) \geq \wp( {\phi}_{AB} , \mathcal{G}) $ $\forall\;\mathcal{G}$ by definition. 
{More precisely,} this is because $\mathbb{E}_\mu[\mathcal{P}]$ is continuous in $\mu$ in $L^1$-norm topology, meaning that its maximizations over $\overline{\mathbb{P}}$ and $\mathbb{P}$ are the same:
\begin{align}\label{Eq: Lemma 1 proof 0001}
{\wp(\rho_{AB},\mathcal{G})=\sup_{\mu \in {{\mathbb{P}}({\rho}_{AB},\mathcal{G})}}\mathbb{E}_{\mu}[\mathcal{P}]=\sup_{\mu \in {\overline{\mathbb{P}}({\rho}_{AB},\mathcal{G})}}\mathbb{E}_{\mu}[\mathcal{P}]}.
\end{align}
Hence, it suffices to show the opposite direction.

Suppose {$
\wp (\rho_{AB}, \mathcal{G}) \geq \wp( {\phi}_{AB} , \mathcal{G}) $ $\forall\;\mathcal{G}$ but} there exists a steering game $\mathcal{G}_0$ such that $\overline{\mathbb{P}}({\rho}_{AB},\mathcal{G}_0)  \not\supseteq \overline{\mathbb{P}}({\phi}_{AB},\mathcal{G}_0)$.
Namely, there is a probability distribution $\mu_* \in \overline{\mathbb{P}}({\phi}_{AB},\mathcal{G}_0)$ such that $\mu_* \notin\overline{\mathbb{P}}({\rho}_{AB},\mathcal{G}_0) $.
Since a steering game can only have its index sets as some Cartesian product of $\mathbb{R}$ and $\mathbb{N}$, which is a space of the form $\mathbb{R}^N\times \mathbb{N}^M$, the set $\mathbb{P}(\cdot)$ {[as defined in Eq.~\eqref{Eq: P set definition}]} must be the closed convex subset of the Banach space $L^1(\mathbb{R}^N\times \mathbb{N}^M)$. 
By applying Hahn-Banach separation theorem~\cite{Zalinescu2002convex}, there exists a bounded linear functional ${\mathcal{P}_*}\in L^\infty(\mathbb{R}^N\times \mathbb{N}^M)$ such that
\begin{align}\label{Eq: Lemma 1 proof 0002}
    \sup_{\mu \in {\overline{\mathbb{P}}({\rho}_{AB},\mathcal{G}_0)}}\mathbb{E}_{\mu}[{\mathcal{P}_*}] <  \mathbb{E}_{\mu_*}[{\mathcal{P}_*}] \leq \sup_{\mu'\in {\overline{\mathbb{P}}({\phi}_{AB},\mathcal{G}_0)}} \mathbb{E}_{\mu'}[{\mathcal{P}_*}].
\end{align}
{Now, let us explicitly write the steering game $\mathcal{G}_0$ as the tuple 
\begin{align}
\mathcal{G}_0=(\mathcal{A},\mathcal{B},\mathcal{X},\mathcal{Y},\{P_\mathcal{X}(x)\}_x,\{\psi_{B'}^y,P_\mathcal{Y}(y)\}_y,\mathcal{P}_0)\eqqcolon(g_0,\mathcal{P}_0), 
\end{align}
where $\mathcal{P}_0$ is its associated pay-off function, and we collectively write $g_0\coloneqq(\mathcal{A},\mathcal{B},\mathcal{X},\mathcal{Y},\{P_\mathcal{X}(x)\}_x,\{\psi_{B'}^y,P_\mathcal{Y}(y)\}_y)$.
Then, one can directly observe from the definition of $\mathbb{P}(\cdot)$ [as given in Eq.~\eqref{Eq: P set definition}] that, for any other steering game of the form $\mathcal{G}=(g_0,\mathcal{P})$, whose only difference with $\mathcal{G}_0$ is the pay-off function, we have
\begin{align}\label{Eq: Lemma 1 proof 0003}
{\overline{\mathbb{P}}({\rho}_{AB},\mathcal{G}_0)}={\overline{\mathbb{P}}({\rho}_{AB},\mathcal{G})}\quad\&\quad{\overline{\mathbb{P}}({\phi}_{AB},\mathcal{G}_0)}={\overline{\mathbb{P}}({\phi}_{AB},\mathcal{G})}.
\end{align}
Now, consider the steering game $\mathcal{G}_*=(g_0,\mathcal{P}_*)$.
Then, using Eqs.~(\ref{Eq: Lemma 1 proof 0001}-\ref{Eq: Lemma 1 proof 0003}), we conclude that
\begin{align}
 \wp(\rho_{AB},\mathcal{G}_*)&=\sup_{\mu \in {\overline{\mathbb{P}}({\rho}_{AB},\mathcal{G}_*)}}\mathbb{E}_{\mu}[{\mathcal{P}_*}]=
    \sup_{\mu \in {\overline{\mathbb{P}}({\rho}_{AB},\mathcal{G}_0)}}\mathbb{E}_{\mu}[{\mathcal{P}_*}]\\
    &<  \mathbb{E}_{\mu_*}[{\mathcal{P}_*}] \leq \sup_{\mu'\in {\overline{\mathbb{P}}({\phi}_{AB},\mathcal{G}_0)}} \mathbb{E}_{\mu'}[{\mathcal{P}_*}]
    =\sup_{\mu'\in {\overline{\mathbb{P}}({\phi}_{AB},\mathcal{G}_*)}} \mathbb{E}_{\mu'}[{\mathcal{P}_*}]=\wp(\phi_{AB},\mathcal{G}_*).
\end{align}
}
This leads to a contradiction and then proves the desired result.
\end{proof}
\bigskip

Now, we detail the proof of Theorem~\ref{Result:steering game combined} in the main text, which is split into two parts for Eqs.~\eqref{theo:steering_monotone} and~\eqref{theo:assemblage_conversion}.\\

\begin{proof}[Proof of {Eq.~\eqref{theo:steering_monotone}}]
Consider an arbitrarily given steering game $\mathcal{G}$ and any given $\mu \in \mathbb{P}({\phi}_{AB}, \mathcal{G})$.
Then, there are POVMs ${E}_A^{a|x}$ and ${M}_{BB'}^{b}$ such that
\begin{align}
    \mu(a,b,x,y) 
    &= {P_\mathcal{X}(x)P_\mathcal{Y}(y)}\mathrm{Tr}\left[\left({E}_A^{a|x}\otimes{M}_{BB'}^{b}\right)\left({\phi}_{AB}\otimes {\psi_{B'}^y}\right)\right] \\
    &={P_\mathcal{X}(x)P_\mathcal{Y}(y)}\mathrm{Tr}\left[{M}_{BB'}^{b}\left({\tilde{\phi}}_{B}^{a|x}\otimes {\psi_{B'}^y}\right)\right],
\end{align}
where $\tilde{\phi}_{B}^{a|x}\coloneqq{\rm Tr}_A\left[\left(E_A^{a|x}\otimes\mathbb{I}_B\right)\phi_{AB}\right]
\in \mathcal{S}_B({\phi}_{AB})
$ is an assemblage induced by ${\phi}_{AB}$.

Now, with the assumption $\rho_{AB}\succeq_{\rm st}\phi_{AB}$, there exists a sequence of LOSR transformations $\{ {\mathcal{E}}_{n}^{a|x}\}_{a,x}$ indexed by $n\in\mathbb{N}$ achieving 
\begin{align}\label{Eq: sup convergence to zero}
    \lim_{n\rightarrow \infty}\sup_{x \in \mathcal{X}}\int_{\mathcal{A}}da\; \left\lVert {\mathcal{E}_{n}^{a|x}({\rho}_{AB})} - {\tilde{\phi}^{a|x}} \right\rVert_1 = 0.
\end{align}
Now, we define the probability distributions
\begin{align}
    \mu_n(a,b,x,y) 
    &\coloneqq P_\mathcal{X}(x)P_\mathcal{Y}(y)\mathrm{Tr}\left[{M}_{BB'}^{b}\left(\mathcal{E}_{n}^{a|x}({\rho}_{AB})\otimes {\psi}_{B'}^y\right)\right],
\end{align}
which satisfies $\mu_n\in\mathbb{P}(\rho_{AB},\mathcal{G})$ for every $n\in\mathbb{N}$ {[see also Eq.~\eqref{Eq: P set definition}]}.
Then it follows that (below, $\norm{\cdot}_{L^1}$ denotes the $L^1$ norm, and $\norm{\cdot}_1$ denotes the trace norm)
\begin{align}
    \lVert  \mu_n -  \mu\rVert_{L^1} 
    &\coloneqq \int_{\mathcal{A},\mathcal{B},\mathcal{X},\mathcal{Y}}da\;db\;dx\;dy\;|\mu_n(a,b,x,y)- \mu(a,b,x,y)| \\
    &= \int_{\mathcal{A},\mathcal{B},\mathcal{X},\mathcal{Y}}da\;db\;dx\;dy\;{P_\mathcal{X}(x)P_\mathcal{Y}(y)}\left| \mathrm{Tr}\left[{M}_{BB'}^{b}{\left(\left(\mathcal{E}_{n}^{a|x}({\rho}_{AB})-\tilde{\phi}^{a|x}\right)\otimes {\psi}_{B'}^y\right)}\right] \right| \\
    &\label{Eq:Computation001}\leq \int_{\mathcal{A},\mathcal{B},\mathcal{X},\mathcal{Y}}da\;db\;dx\;dy\;{P_\mathcal{X}(x)P_\mathcal{Y}(y)}\mathrm{Tr}\left[{M}_{BB'}^{b}{\left(\left|\mathcal{E}_{n}^{a|x}({\rho}_{AB})-\tilde{\phi}^{a|x}\right|\otimes {\psi}_{B'}^y\right)}\right] \\
    &= \int_{\mathcal{X}}dx\; {P_\mathcal{X}(x)} \int_{\mathcal{A}}da\; \left\lVert {\mathcal{E}_{n}^{a|x}( {\rho}_{AB})-\tilde{\phi}^{a|x}}\right\rVert_1 \\
    &\leq \int_{\mathcal{X}}dx\; P_\mathcal{X}(x)\sup_{x'\in\mathcal{X}} \int_{\mathcal{A}}da\; \left\lVert \mathcal{E}_{n}^{a|x'}( {\rho}_{AB})-\tilde{\phi}^{a|x'}\right\rVert_1\\
    &= {\sup_{x'\in\mathcal{X}} \int_{\mathcal{A}}da\; \left\lVert \mathcal{E}_{n}^{a|x'}( {\rho}_{AB})-\tilde{\phi}^{a|x'}\right\rVert_1}.
\end{align}
In Eq.~\eqref{Eq:Computation001}, we use the fact that $|{\rm Tr}(PQ)|\le{\rm Tr}(P|Q|)$ for every positive semi-definite operator $P\succeq0$ and hermitian operator $Q$.
{Using Eq.~\eqref{Eq: sup convergence to zero}, we thus conclude that
\begin{align}
\lim_{n\rightarrow \infty}\lVert  \mu_n -  \mu\rVert_{L^1}=0.
\end{align}
In other words, since $\mu_n\in{\mathbb{P}({\rho}_{AB},\mathcal{G})}$ $\forall\,n$,
\begin{align}
\inf_{\mu'\in {\mathbb{P}({\rho}_{AB},\mathcal{G})}}\lVert \mu' -\mu \rVert_{L^1} =0.
\end{align}
Note that the above argument works for arbitrarily given steering game $\mathcal{G}$ and $\mu \in \mathbb{P}({\phi}_{AB}, \mathcal{G})$.}
As a result, 
for {every} steering games $\mathcal{G}$, we have
\begin{align}
 \inf_{\mu'\in {\mathbb{P}({\rho}_{AB},\mathcal{G})}}\lVert \mu' -\mu \rVert_{L^1} =0 \quad\forall\;\mu \in \mathbb{P}({{\phi}_{AB}},\mathcal{G}).
\end{align}
Hence, we conclude that
\begin{align}
    \overline{\mathbb{P}}({\rho}_{AB},\mathcal{G}) \supseteq {\mathbb{P}}({\phi}_{AB},\mathcal{G})\quad\forall\; \mathcal{G},
\end{align}
where we recall that $\overline{\mathbb{P}}$ is the closure in the $L^1$-norm topology.
Since $\overline{\mathbb{P}}({\rho}_{AB},\mathcal{G})$ is closed in this topology, this also implies that
\begin{align}
  \overline{\mathbb{P}}({\rho}_{AB},\mathcal{G}) \supseteq \overline{\mathbb{P}}({\phi}_{AB},\mathcal{G})\quad\forall\; \mathcal{G},  
\end{align}
which is the desired result $\rho_{AB}\succeq_{\rm game}\phi_{AB}$.
\end{proof}

\bigskip

{Finally, to complete the proof of Theorem~\ref{Result:steering game combined} in the main text, it remains to prove Eq.~\eqref{theo:assemblage_conversion}.}\\

\begin{proof}[Proof of {Eq.~\eqref{theo:assemblage_conversion}}]
With the assumption $\rho_{AB}\succ_{\rm game}\phi_{AB}$, which means that ${\mathbb{P}}({\rho}_{AB},\mathcal{G}) \supseteq \overline{\mathbb{P}}({\phi}_{AB},\mathcal{G})$ $\forall\; \mathcal{G}$, we consider a special type of steering games that have the index set of $\mathcal{Y} = \mathbb{R}^{2N}$, where $N$ is now the number of modes for the state $\mathrm{Tr}_A[\phi_{AB}]$ (not the number operator anymore).
Then, by setting $y=(y_x^{(1)},y_p^{(1)},y_x^{(2)},y_p^{(2)},...,y_x^{(N)},y_p^{(N)})$, we can take the $N$-mode coherent states $\{|\alpha_{y}\rangle \langle \alpha_{y}| \}_{y\in\mathbb{R}^{2N}}$ as the inputs $\psi_{B'}^y$ for Bob's side, {where [also recall from Eq.~\eqref{Eq:coherent state definition}]}
\begin{align}\label{Eq:N-mode coherent state}
    |\alpha_{y}\rangle  
    \coloneqq{{{D}}(\alpha_{y})\ket{0}\coloneqq} \bigotimes_{n = 1}^Ne^{i \sqrt{2}y_p^{(n)} \hat{x}_n-i\sqrt{2}y_x^{(n)}\hat{p}_n} |0\rangle
    =\bigotimes_{n = 1}^N |y_x^{{(n)}} + iy_p^{{(n)}}\rangle,
\end{align}
{
where  the $N$-mode displacement operator is defined as
\begin{align}\label{Eq:N-mode displacement operator}
{{D}}(\alpha_{y})\coloneqq \bigotimes_{n = 1}^Ne^{i \sqrt{2}y_p^{(n)} \hat{x}_n-i\sqrt{2}y_x^{(n)}\hat{p}_n}.
\end{align}
For any} steering game of this type (which is again denoted by $\mathcal{G}$), we {again} have ${\mathbb{P}}({\rho}_{AB},\mathcal{G}) \supseteq \overline{\mathbb{P}}({\phi}_{AB},\mathcal{G})$ {from the assumption $\rho_{AB}\succ_{\rm game}\phi_{AB}$}.
Now, for arbitrarily given POVMs $\{{M}_{BB'}^{b}\}_b$, $\{ {E}_{A}^{a|x}\}_a$, and channel $\mathcal{L}_B$ acting on $B$,
we have (below, ${\rm id}_A$ denotes the identity map acting on $A$, {and $P_\mathcal{X},P_\mathcal{Y}$ are the probability distributions on $\mathcal{X},\mathcal{Y}$ contained in the steering game $\mathcal{G}$})
\begin{align}
   & {P_\mathcal{X}(x)}P_\mathcal{Y}(y)\mathrm{Tr}\left[ \left({E}_A^{a|x}\otimes{M}_{BB'}^{b}\right)\Big( \left({\rm id}_A\otimes \mathcal{L}_B\right)({\phi}_{AB}) \otimes \proj{\alpha_{y}}_{B'} \Big)\right] \notag\\
    &\quad\quad= {P_\mathcal{X}(x)}P_\mathcal{Y}(y)\mathrm{Tr}\left[ \left({E}_A^{a|x}\otimes(\mathcal{L}_B^\dagger \otimes {\rm id}_{B'})({M}_{BB'}^{b})\right)({\phi}_{AB} \otimes \proj{\alpha_{y}}_{B'})\right],
\end{align}
where $\{(\mathcal{L}_B^\dagger \otimes {\rm id}_{B'})({M}_{BB'}^{b})\}_b$ is again a POVM acting on $BB'$.
This means the above probability distribution is in ${\mathbb{P}}({\phi}_{AB},\mathcal{G})$.
Since ${\mathbb{P}}({\rho}_{AB},\mathcal{G}) \supseteq \overline{\mathbb{P}}({\phi}_{AB},\mathcal{G})$, we conclude that there exist POVMs $\{{R}_{BB'}^{b|\lambda}\}_a$ and $\{ {Z}_{A}^{a|x,\lambda} \}_a$ and shared randomness $\lambda\in\Lambda$ such that
\begin{align}\label{Eq: proof equality 001}
    &{P_\mathcal{X}(x)}P_\mathcal{Y}(y)\mathrm{Tr}\left[ \left({E}_A^{a|x}\otimes{M}_{BB'}^{b}\right)\Big( \left({\rm id}_A\otimes \mathcal{L}_B\right)({\phi}_{AB}) \otimes \proj{\alpha_{y}}_{B'} \Big)\right]\nonumber\\
    &\quad\quad={P_\mathcal{X}(x)}P_\mathcal{Y}(y)\int_{\Lambda}d\mu(\lambda)\mathrm{Tr}\left[  \left({Z}_{A}^{a|x,{\lambda}}\otimes{R}_{BB'}^{b|\lambda}\right)( {{\rho}_{AB}} \otimes \proj{\alpha_{y}}_{B'} )\right].
\end{align}
Now, {to proceed,} we rewrite the two-mode squeezed state as defined in Eq.~\eqref{eq:TMSS} into the following form for some $\beta\in(0,\infty)$ [here, we use a different notation to avoid confusion with Eq.~\eqref{eq:TMSS}]:
\begin{align}
    |{\Phi^\beta_{BB'}} \rangle = {\sqrt{1-e^{-2\beta}}} {\sum_{n=0}^\infty} e^{-n\beta}{\ket{nn}_{BB'}},
\end{align}
which satisfies the following equality for every coherent state $\ket{\alpha}$ {[defined in Eq.~\eqref{Eq:coherent state definition}]}:
\begin{align}
    | \alpha \rangle_{B'} =  \frac{e^{\frac{|\alpha|^2}{2}(e^{2\beta}-1)}}{\sqrt{1-e^{-2\beta}}} \left( \langle {e^{\beta }\alpha} |_{B} \otimes {\mathbb{I}_{B'}} \right) |{\Phi^\beta_{BB'}} \rangle.
\end{align}
Hence, by redefining $e^\beta\alpha_{y}\to\alpha_{y}$ and set {$P_\mathcal{X}(x)>0$ $\forall\,x\in\mathcal{X}$ as well as} $P_\mathcal{Y}(y)> 0$ $\forall\,y\in\mathcal{Y}$, Eq.~\eqref{Eq: proof equality 001} can be further rewritten into (below, $B''$ is a system identical to $B$ and $B'$, and we write $\Phi^\beta\coloneqq\proj{\Phi^\beta}$)
\begin{align}
   &\mathrm{Tr}\left[ \left( {E}_A^{a|x}\otimes{M}_{BB'}^{b}\otimes\proj{\alpha_{y}}_{B''}\right)\left( \left({\rm id}_A\otimes \mathcal{L}_B\right)({\phi}_{AB})\otimes {{(\Phi^\beta)^{\otimes N}_{B'B''}}} \right)\right]\nonumber\\ 
    &\quad\quad=  \int_{\Lambda}d\mu(\lambda)\mathrm{Tr}\left[ \left( {Z}_A^{a|x,{\lambda}}\otimes{R}_{BB'}^{b|\lambda}\otimes\proj{\alpha_{y}}_{B''}\right)\left( {\rho_{AB}} \otimes {{(\Phi^\beta)^{\otimes N}_{B'B''}}} \right)\right].
\end{align}
Because the Husimi-Q function of a state [defined as $\mathcal{Q}_{\phi} (\alpha)\coloneqq \langle \alpha| \phi |\alpha\rangle$ for coherent states $\ket{\alpha}$] is an unique representation~\cite{Serafini2023}, and we have set $\ket{\alpha_{y}}$'s to be coherent states, we conclude that
\begin{align}\label{Eq: proof computation 002}
    &\mathrm{Tr}_{ABB'}\left[ \left( {E}_A^{a|x}\otimes{M}_{BB'}^{b}\otimes{\mathbb{I}_{B''}}\right)\left({\left({\rm id}_A\otimes \mathcal{L}_B\right)({\phi}_{AB})}  \otimes {(\Phi^\beta)^{\otimes N}_{B'B''}} \right)\right] \nonumber\\ 
    &\quad\quad=\int_{\Lambda}d\mu(\lambda)\mathrm{Tr}_{ABB'}\left[ \left({Z}_A^{a|x,\lambda}\otimes{R}_{BB'}^{b|\lambda}\otimes{\mathbb{I}_{B''}}\right)\left( {{\rho}_{AB}} \otimes {(\Phi^\beta)^{\otimes N}_{B'B''}} \right)\right].
\end{align}

Now, let us set {$\mathcal{B}=\mathbb{R}^{2N}$ and write $b\in\mathcal{B}=\mathbb{R}^{2N}$ as} $b=(b_x^{(1)},b_p^{(1)},b_x^{(2)},b_p^{(2)},...,b_x^{(N)},b_p^{(N)})$.
{Then, by using the $N$-mode displacement operator ${{D}}(\alpha_{b})$ given in Eq.~\eqref{Eq:N-mode displacement operator} for the variable $b$, namely,}
%
\begin{align}
{{D}}(\alpha_{b})\coloneqq \bigotimes_{n = 1}^Ne^{i \sqrt{2}b_p^{(n)} \hat{x}_n-i\sqrt{2}b_x^{(n)}\hat{p}_n},
\end{align}
{we can write} the $2N$-mode CV Bell measurements as 
\begin{align}
{\rm BSM}_{BB'}^b\coloneqq\frac{1}{(2\pi)^N}({{D}}_{B}(\alpha_{b})\otimes {\mathbb{I}_{B'}}){\proj{\Psi_{BB'}^{(N)}}}({{D}}_{B}(\alpha_{b})\otimes {\mathbb{I}_{B'}})^\dag.
\end{align}
{Here,} $|\Psi^{(N)}_{BB'}\rangle \coloneqq \sum_{\bm{n} \in \mathbb{N}_0^N} \ket{\bm{n}}_B\otimes\ket{\bm{n}}_{B'}$ is the unnormalized maximum entangled state, and $\mathbb{N}_0\coloneqq\{0\}\cup\mathbb{N}$.
We then set ${M}_{BB'}^{b}$ to be this CV Bell measurements, that is,
\begin{align}
    {M}_{BB'}^{b}& = {{\rm BSM}_{BB'}^b}.
\end{align}

{Now, we write the assemblage induced by $\phi_{AB}$ as
\begin{align}\label{Eq: proof assemblage of phi}
    \tilde{\phi}^{a|x}_B 
    \coloneqq \mathcal{L}_B\circ\mathrm{Tr}_{A}\left[ \left( {E}_A^{a|x}\otimes \mathbb{I}_{B}\right){\phi}_{AB}  \right]. 
\end{align}
Then, Eq.~\eqref{Eq: proof computation 002} becomes} (below, $\hat{n}$ is the $N$-mode number operator, and, in this proof, we use lower case to avoid confusion with the number of modes $N$)
\begin{align}
    \left( \frac{1-e^{-2\beta}}{2\pi} \right)^N e^{-\beta \hat{n}}{{D}}_{B''}^\dag(\alpha_{b}) {\tilde{\phi}^{a|x}_{B''}}{{D}}_{B''}(\alpha_{b})e^{-\beta \hat{n}} = 
    \int_{\Lambda}d\mu(\lambda)\mathrm{Tr}_{ABB'}\left[ \left({Z}_A^{a|x,\lambda}\otimes{R}_{BB'}^{b|\lambda}\otimes{\mathbb{I}_{B''}}\right)\left( {\rho_{AB}} \otimes {(\Phi^\beta)^{\otimes N}_{B'B''}}\right)\right]. 
\end{align}
Now, we apply the displacement operator ${{D}}_{B''}(\alpha_{b})$ {and ${{D}}_{B''}^\dagger(\alpha_{b})$} on both sides and take an integral over $b$, resulting in
\begin{align}
   {\mathcal{I}}_{{B''}}^{\beta}\left( {\tilde{\phi}^{a|x}_{B''}}\right) =\int_{\Lambda}d\mu(\lambda)\int_{\mathbb{R}^{2N}}db \;{{D}}_{B''}(\alpha_{b})\mathrm{Tr}_{ABB'}\left[ \left({Z}_A^{a|x,\lambda}\otimes{R}_{BB'}^{b|\lambda}\otimes{\mathbb{I}_{B''}}\right)\left({\rho_{AB}} \otimes {(\Phi^\beta)^{\otimes N}_{B'B''}} \right)\right]{{D}}_{B''}^\dag(\alpha_{b}),\;\label{eq:conversion_0}
\end{align}
where the {mapping} ${\mathcal{I}}_{{B''}}^{\beta}$ is the random displacement channel acting on the system $B''$ defined as (see, e.g., Ref.~\cite{Lami2018CVtopology})
\begin{align}\label{Eq:random displacement channel def}
  {\mathcal{I}}_{{B''}}^{\beta}\left( \cdot \right)\coloneqq&  \left( \frac{1-e^{-2\beta}}{2\pi} \right)^N\int_{\mathbb{R}^{2N}}db\; {{D}}_{B''}(\alpha_{b})e^{-\beta \hat{n}}{{D}}_{B''}^\dag(\alpha_{b}){( \cdot )}{{D}}_{B''}(\alpha_{b})e^{-\beta \hat{n}}{{D}}_{B''}^\dag(\alpha_{b}) \notag\\
  =&\frac{1}{\left(2\pi  \tanh(\frac{\beta}{2})\right)^N}\int_{\mathbb{R}^{2N}}db \; e^{-\frac{|b|^2}{2\tanh(\frac{\beta}{2})}}{{D}}_{B''}(\alpha_{b}){( \cdot )}{{D}}_{B''}^\dag(\alpha_{b}).
\end{align}
By defining
\begin{align}
    \mathcal{L}_{B''}^{\beta,\lambda}(\cdot)\coloneqq\int_{\mathbb{R}^{2N}}db \;{{D}}_{B''}(\alpha_{b})\mathrm{Tr}_{BB'}\left[ \left({R}_{BB'}^{b|\lambda}\otimes{\mathbb{I}_{B''}}\right){\left( (\cdot)_B \otimes {(\Phi^\beta)^{\otimes N}_{B'B''}} \right)}\right]{{D}}_{B''}^\dag(\alpha_{b}),
\end{align}
which is a channel acting on $B''$, a LOSR transformation (with $\beta$ dependence) can therefore be given by
\begin{align}\label{Eq: LOSSR transformation without mixture}
    \mathcal{E}^{\beta,a|x}(\cdot)\coloneqq\int_{\Lambda}d\mu(\lambda)\;{\mathcal{L}^{\beta,\lambda}_{B''}\circ}\mathrm{Tr}_{A}\left[ \left({Z}_A^{a|x,\lambda}\otimes{\mathbb{I}_{B}}\right){(\cdot)_{AB}}\right].
\end{align}
Hence, Eq.~\eqref{eq:conversion_0} can be rewritten as
\begin{align}
    {\mathcal{I}}^\beta_{B''}\left( \tilde{\phi}^{a|x}_{B''}\right) = {\mathcal{E}}^{\beta,a|x}({\rho}_{AB}) \label{eq:conversion_1}.
\end{align}
An important observation is that the LOSR transformation ${\mathcal{E}}^{\beta,a|x}$ on the right-hand side depends on $\Lambda$, $Z_A^{a|x,\lambda}$, and $R_{BB'}^{b|\lambda}$, while the random displacement channel ${\mathcal{I}}^\beta_{B''}$ {\em does not}.


At this point, we obtain Eq.~\eqref{eq:conversion_1} with arbitrarily given POVMs $\{E_A^{a|x}\}_a$ and channel $\mathcal{L}_B$.
Now, we further consider the probabilistic mixture between them, which is extending Eq.~\eqref{Eq: proof assemblage of phi} into
\begin{align}\label{Eq: proof assemblage of sigma}
    \tilde{\sigma}_{{B}}^{a|x}
    &\coloneqq  \int_{\Lambda'}d\mu(\lambda')\;{\mathcal{L}_B^{\lambda'}\circ\mathrm{Tr}_{A}\left[ \left( {E}_A^{a|x,\lambda'}\otimes\mathbb{I}_{B}\right){\phi}_{AB}  \right]} \\
    &\eqqcolon \int_{\Lambda'}d\mu(\lambda')\;\tilde{\phi}_{{B}}^{a|x,\lambda'},
\end{align}
where we define 
\begin{align}
\tilde{\phi}_{B}^{a|x,\lambda'}\coloneqq\mathcal{L}_B^{\lambda'}\circ\mathrm{Tr}_{A}\left[ \left( {E}_A^{a|x,\lambda'}\otimes\mathbb{I}_{B}\right){\phi}_{AB}  \right], 
\end{align}
which is basically Eq.~\eqref{Eq: proof assemblage of phi} with $\lambda'$-dependence. This also means that, for each $\lambda'$, {the mapping} $\mathcal{L}_B^{\lambda'}$ is a channel and $\{{E}_A^{a|x,\lambda'}\}_{a,x}$ is a set of POVMs.
Then, we can repeat the same argument for every $\tilde{\phi}_{B''}^{a|x,\lambda'}$.
For each $\lambda'$ and every $\beta\in(0,\infty)$, we obtain an LOSR transformation of the form Eq.~\eqref{Eq: LOSSR transformation without mixture}, but now with $\lambda'$-dependence and is denoted as $\mathcal{E}^{\beta,a|x,\lambda'}$. 
They achieve the relation in Eq.~\eqref{eq:conversion_1}; namely, for every $\beta\in(0,\infty)$, 
\begin{align}
    {\mathcal{I}}^\beta_{B''}\left( \tilde{\phi}^{a|x,\lambda'}_{B''}\right) = {\mathcal{E}}^{\beta,a|x,\lambda'}({\rho}_{AB}).
\end{align}
Crucially, note that the random displacement channel $\mathcal{I}_{B''}^\beta$ does not depend on $\lambda'$. This means that
\begin{align}
    \mathcal{I}^{\beta}_{{B''}}(\tilde{\sigma}_{B''}^{a|x}) = \int_{\Lambda'}d\mu(\lambda')\; \mathcal{I}^{\beta}_{{B''}} (\tilde{\phi}_{B''}^{a|x,\lambda'})=\int_{\Lambda'}d\mu(\lambda')\;\mathcal{E}^{\beta,a|x,\lambda'}(\rho_{AB}),
\end{align}
which is a convex mixture of LOSR transformations acting on $\rho_{AB}$.
Define
\begin{align}
\tilde{\rho}^{\beta,a|x}_{B''}\coloneqq\int_{\Lambda'}d\mu(\lambda')\;\mathcal{E}^{\beta,a|x,\lambda'}(\rho_{AB}).
\end{align}
{Then, by definition given in Eq.~\eqref{Eq:LOSR induced state assemblages}, we have} $\tilde{\rho}^{\beta,a|x}_{B''}\in\mathcal{S}_B(\rho_{AB})$ for every $\beta\in(0,\infty)$.
Now, we recall the fact that the random displacement channel {defined in Eq.~\eqref{Eq:random displacement channel def}} converges to the identity channel in the so-called strong operator topology~\cite{Lami2018CVtopology}, that is, for the given $\phi_{AB}$,
\begin{align}\label{Eq: random displacement channel to identity}
    \lim_{\beta \rightarrow 0^+}\left\lVert(\mathrm{id}_A\otimes \mathcal{L}^{\lambda'}_{B})(\phi_{AB})- \big(\mathrm{id}_A\otimes(\mathcal{I}^{\beta}\circ \mathcal{L}^{\lambda'})_{B}\big)(\phi_{AB})\right\rVert_1=0\quad\forall\; \lambda'\in\Lambda'.
\end{align}
This means there is a sequence $\{\beta_n\}_{n=1}^\infty$ such that $\lim_{n\to\infty}\beta_n=0$ and 
\begin{align}\label{Eq: random displacement channel to identity002}
\lim_{n\to\infty}\left\lVert(\mathrm{id}_A\otimes \mathcal{L}^{\lambda'}_{B})(\phi_{AB})- \big(\mathrm{id}_A\otimes (\mathcal{I}^{\beta_n}\circ \mathcal{L}^{\lambda'})_{B}\big)(\phi_{AB})\right\rVert_1=0\quad\forall\;\lambda'\in\Lambda'.
\end{align}
Then, for the assemblage $\tilde{\sigma}^{a|x}$ defined in Eq.~\eqref{Eq: proof assemblage of sigma} and the given $\rho_{AB}$, 
a direct computation shows that (recall again that $B''$ is a system identical to $B$)
\begin{align}
  &\inf_{\tilde{\rho}^{a|x} \in \mathcal{S}_B({\rho}_{AB})}
   \sup_{x\in \mathcal{X}}\int_{\mathcal{A}}da\; \left\lVert \tilde{\sigma}^{a|x}_{{B''}} - \tilde{\rho}^{a|x}_{{B''}} \right\rVert_1
    \leq\lim_{{n\to\infty}}\sup_{x\in \mathcal{X}}\int_{\mathcal{A}}da\;\left\lVert \tilde{\sigma}^{a|x}_{{B''}}-  \tilde{\rho}^{{\beta_n},a|x}_{{B''}}\right\rVert_1 \\
    &\qquad\qquad= \lim_{{n\to\infty}}\sup_{x\in \mathcal{X}}\int_{\mathcal{A}}da\;\left\lVert \tilde{\sigma}^{a|x}_{{B''}} -  {\mathcal{I}}^{{\beta_n}}_{{B''}}\left( \tilde{\sigma}^{a|x}_{{B''}}\right)\right\rVert_1 \\
    &\qquad\qquad= \lim_{{n\to\infty}}\sup_{x\in \mathcal{X}}\int_{\mathcal{A}}da\;\left\lVert\int_{\Lambda'}d\mu(\lambda') (\mathcal{L}^{\lambda'}-\mathcal{I}^{{\beta_n}}\circ \mathcal{L}^{\lambda'})_{B''}\circ\mathrm{Tr}_A\left[\left({E}_A^{a|x,\lambda'}\otimes \mathbb{I}_{B''}\right)\phi_{AB''} \right] \right\rVert_1 \\
    &\qquad\qquad\leq \lim_{{n\to\infty}}\sup_{x\in \mathcal{X}}\int_{\mathcal{A}}da\int_{\Lambda'}d\mu(\lambda')\;\left\lVert (\mathcal{L}^{\lambda'}-\mathcal{I}^{{\beta_n}}\circ \mathcal{L}^{\lambda'})_{B''}\circ\mathrm{Tr}_A\left[\left({E}_A^{a|x,\lambda'}\otimes \mathbb{I}_{B''}\right)\phi_{AB''} \right] \right\rVert_1 \\
    &\qquad\qquad=\lim_{{n\to\infty}}\sup_{x\in \mathcal{X}}\int_{\mathcal{A}}da\int_{\Lambda'}d\mu(\lambda')\;\mathrm{Tr}\left| (\mathcal{L}^{\lambda'}-\mathcal{I}^{\beta_n}\circ \mathcal{L}^{\lambda'})_{B''}\circ\mathrm{Tr}_A\left[\left({E}_A^{a|x,\lambda'}\otimes \mathbb{I}_{B''}\right)\phi_{AB''} \right]\right| \\
    &\qquad\qquad\leq \lim_{{n\to\infty}}\sup_{x\in \mathcal{X}}\int_{\mathcal{A}}da\int_{\Lambda'}d\mu(\lambda')\;\mathrm{Tr}\left[ 
   \left({E}_A^{a|x,\lambda'}\otimes \mathbb{I}_{B''}\right)\left|    \big(\mathrm{id}_A\otimes(\mathcal{L}^{\lambda'}-\mathcal{I}^{{\beta_n}}\circ \mathcal{L}^{\lambda'})_{B''}\big)(\phi_{AB''})\right|\right]\\
    &\qquad\qquad= \lim_{{n\to\infty}}\int_{\Lambda'}d\mu(\lambda')\;\left\lVert(\mathrm{id}_A\otimes \mathcal{L}^{\lambda'}_{B''})(\phi_{AB''})- \big(\mathrm{id}_A\otimes (\mathcal{I}^{{\beta_n}}\circ \mathcal{L}^{\lambda'})_{B''}\big)(\phi_{AB''})\right\rVert_1,\label{eq:conversion_result}
\end{align}
{where again we use the fact that $|{\rm Tr}(PQ)|\le{\rm Tr}(P|Q|)$ for every positive semi-definite operator $P\succeq0$ and hermitian operator $Q$.}
Now, define the sequence of functions
\begin{align}
g_n(\lambda')\coloneqq\left\lVert(\mathrm{id}_A\otimes \mathcal{L}^{\lambda'}_{B''})(\phi_{AB''})- \big(\mathrm{id}_A\otimes (\mathcal{I}^{\beta_n}\circ \mathcal{L}^{\lambda'})_{B''}\big)(\phi_{AB''})\right\rVert_1.
\end{align}
Then, by Eq.~\eqref{Eq: random displacement channel to identity002}, we have $\lim_{n\to\infty}g_n(\lambda')=0$ for every $\lambda'\in\Lambda'$ and, by the triangle inequality of $\norm{\cdot}_1$, we have that
\begin{align}
g_n(\lambda')\le\left\lVert(\mathrm{id}_A\otimes \mathcal{L}^{\lambda'}_{B''})(\phi_{AB''})\right\rVert_1+\left\lVert\big(\mathrm{id}_A\otimes (\mathcal{I}^{\beta_n}\circ \mathcal{L}^{\lambda'})_{B''}\big)(\phi_{AB''})\right\rVert_1\le2\quad\forall\;\lambda'\in\Lambda'\;\&\;n\in\mathbb{N}.
\end{align}
Also, note that the function $h(\lambda')\equiv2$ is (Lebesgue) integrable on $\Lambda'$, that is, $\int_{\Lambda'}d\mu(\lambda')h(\lambda')=2\int_{\Lambda'}d\mu(\lambda')=2<\infty$ [recall that $d\mu(\lambda')$ is a probability measure on $\lambda'\in\Lambda'$].
Then, by applying Lebesgue’s dominated convergence theorem~\cite{Real_Analysis} with $h(\lambda')$ as the dominated function, we conclude that
\begin{align}
    \inf_{\tilde{\rho}^{a|x} \in \mathcal{S}_B({\rho}_{AB})}
   \sup_{x\in \mathcal{X}}\int_{\mathcal{A}}da\; \left\lVert \tilde{\sigma}^{a|x}_{{B''}} - \tilde{\rho}^{a|x}_{{B''}} \right\rVert_1
   \leq& \lim_{n\to\infty}\int_{\Lambda'}d\mu(\lambda')\;g_n(\lambda')=\int_{\Lambda'}d\mu(\lambda')\;\lim_{n\to\infty}g_n(\lambda')=0.
\end{align}
Note that the above argument holds for every $\tilde{\sigma}^{a|x} \in \mathcal{S}_B(\phi_{AB})$, as Eq.~\eqref{Eq: proof assemblage of sigma} is the expression of a generic element in $\mathcal{S}_B(\phi_{AB})$.
This thus means ${\rho_{AB}} \succeq_{\mathrm{st}} {\phi_{AB}}$, 
which completes the proof.
\end{proof}

\section{Bipartite CV Steering Criteria}
\label{app:CV_steer}
\subsection{{Reid's Criterion}}
In the seminal paper~\cite{reid1989demonstration}, Reid proposed a noisy version of the original Einstein-Podolsky-Rosen paradox~\cite{EPR} as a witness of steering. Consider the conditioned state of Bob after the measurement of Alice. Alice has two measurement choices, which we will label as $x$ and $p$, obtaining as a result the random variable $v_x$ or $v_p$. Knowing the value of her result, Alice can try to estimate the corresponding Bob's value of $\hat{x}_B$, or $\hat{p}_B$. We can therefore build Alice's error operators as follows:
\begin{align}
    \hat{x}_{\rm err}&\coloneqq\hat{x}_B|_{v_x}-x_{\rm est}(v_x)\;, \\
    \hat{p}_{\rm err}&\coloneqq\hat{p}_B|_{v_p}-p_{\rm est}(v_p)\;,
\end{align}
Reid's criterion says that unsteerable states satisfy the Heisenberg uncertainty principle for $\hat{x}_{\rm err}$ and $\hat{p}_{\rm err}$. That is, the following inequality holds for unsteerable states:
\begin{align}
    \langle \hat{x}_{\rm err}^2\rangle \langle \hat{p}_{\rm err}^2\rangle \geq \frac{1}{4},
    \label{eq:uncert_prod}
\end{align}
which is based on the standard Heisenberg uncertainty relation.
Similarly, in such a case, we also have
\begin{align}
    \langle \hat{x}_{\rm err}^2\rangle + \langle \hat{p}_{\rm err}^2\rangle \geq 1\;.
    \label{eq:uncert_sum}
\end{align}
Notice that $\langle\hat{x}_{\rm err}^2\rangle$ is minimized when $x_{\rm est}(v_x)=\langle\hat{x}_B\rangle|_{v_x}$, and similarly for $\hat{p}_{\rm err}$. That is, the optimal estimation by Alice is that of guessing the average value of Bob's quadratures conditioned on Alice's outcome. In such a case, one has
\begin{align}
    \langle\hat{x}_{\rm err}^2\rangle \equiv {\rm Var}[\hat{x}_B]|_{v_x}\;,\quad 
     \langle\hat{p}_{\rm err}^2\rangle \equiv {\rm Var}[\hat{p}_B]|_{v_p}\;.
\end{align}

\subsection{Covariance Matrix Criterion}
An equivalent criterion is that proposed by Wiseman~\emph{et~al.} in Ref.~\cite{wiseman2007steering} in terms of the {\em covariance matrix} of $\rho_{AB}$ which is a matrix defined as $(V_{{AB}})_{ij}\coloneqq\langle\{\Delta \hat{r}_i,\Delta \hat{r}_j\}\rangle_{\rho_{AB}}$, where $\hat{r}_i$ is the vector $\{\hat{x}_A,\hat{p}_A,\hat{x}_B,\hat{p}_B\}$.
As a matrix, $V_{AB}$ can be expressed as
\begin{align}
    V_{AB}=\begin{pmatrix}
        V_A & C \\ C^\intercal & V_B
    \end{pmatrix}
\end{align}
Then, the condition states that a Gaussian
$\rho_{AB}$ is unsteerable if and only if
\begin{align}
    V_{AB}+ \begin{pmatrix}
        0 & 0 \\ 0 & i\Omega
    \end{pmatrix}
    \geq 0.
\end{align}
Here, $\Omega$ is the usual symplectic matrix $\Omega \coloneqq \begin{pmatrix} 0 & 1 \\ -1 & 0 \end{pmatrix}$.

\subsection{The Case of a Two-Mode Squeezed State}
\label{sec:TMSS}
Consider a pure two-mode squeezed state as defined in Eq.~\eqref{eq:TMSS}, whose covariance matrix can be expressed as
\begin{align}
    V_{AB}(\gamma)=\begin{pmatrix}
        \cosh2{\gamma} & 0 & \sinh{2\gamma} & 0 \\
        0 & \cosh{2\gamma} & 0 & -\sinh{2\gamma} \\
        \sinh{2\gamma} & 0 & \cosh{2\gamma} & 0 \\
        0 & -\sinh{2\gamma} & 0 & \cosh{2\gamma} \\
    \end{pmatrix}\;,
\end{align}
where $\gamma$ is the squeezing parameter. 
Then, it follows that upon Alice measuring $\hat{x}$, the marginal distribution on $\hat{x}_B$ is Gaussian with the variance
\begin{align}
    {\rm Var}[\hat{x}_B]|_{v_x}=\frac{1}{2\cosh 2\gamma}\;.
\end{align}
In a similar fashion, one obtains
\begin{align}
    {\rm Var}[\hat{p}_B]|_{v_p}=\frac{1}{2\cosh 2\gamma}\;.
\end{align}
By comparing the above expressions with the uncertainty principles characterized by Eqs.~\eqref{eq:uncert_prod} and~\eqref{eq:uncert_sum}, we see that all pure two-mode squeezed states are steerable with Gaussian measurements as soon as the squeezing $\gamma$ is nonzero.

\section{Technical Details for MDI Detection with Gaussian States and Measurements (Sec.~\ref{sec:MDI_GAUSS})}

\subsection{Proof of The Unsteerable Bound}
\label{app:MDI_gauss_wit}

If $\rho_{AB}$ is unsteerable, then, whatever the measurement of Alice, the resulting Alice-to-Bob assemblage is of the form
\begin{align}
    \rho_B^{a|r}=\int d\lambda f(\lambda) p(a|r,\lambda)\sigma^{(\lambda)}_B.
\end{align}
Therefore, the statistics of $a,b_1,b_2$ will be given by
\begin{align}
    p(a,b_1,b_2|r,\ket{\beta})=\Tr\left[M^{b_1,b_2}_{BB'}{\left(\rho^{{a|r}}_B\otimes\ketbra{\beta}_{B'}\right)}\right]
\end{align}
for some local POVM $M^{b_1,b_2}_{BB'}$.
We then have
\begin{align}
    \left\langle \left( b_1 - \frac{a}{\sqrt{2}}-\beta_x\right)^2\right\rangle\bigg|_{r=x}
    =\int d\lambda f(\lambda) \left\langle \left( b_1 - \frac{a}{\sqrt{2}}-\beta_x\right)^2\right\rangle\bigg|_{r=x,\lambda}
\end{align}
and 
a similar form can also be obtained
for the second term in Eq.~\eqref{eq:witness}. Its minimum can be therefore lower bounded by fixing a single value of $\lambda$, that is, assuming that there is no randomness associated. 
Hence, without loss of generality, we can look for the minimum score associated to a single value of $\lambda$ i.e. substitute $f(\lambda)\equiv \delta(\lambda-\lambda_0)$ and therefore we can assume uncorrelated $a$ and $\vec{b}=(b_1,b_2)$.
Then, the assemblage becomes $\rho_B^{a|r}=p(a|r)\sigma_B$, and consequently $\sigma_B\equiv\rho_B$. It is clear that in such a case, as $a$ is uncorrelated to $\Vec{b}$ and $\ket{\beta}$, the witness is lower bounded by
\begin{align}
    \mathcal{W}\geq {\rm Var }[a]+\left\langle (b_1-\beta_x)^2+(b_2-\beta_p)^2\right\rangle|_{\rho_B}
\end{align}
and, therefore,
\begin{align}
    \mathcal{W}\geq \left\langle (b_1-\beta_x)^2+(b_2-\beta_p)^2\right\rangle|_{\rho_B}
\end{align}
in the unsteerable case.
From standard estimation theory inequalities, it follows that if the prior distribution of $\beta$ is wide enough, this lower bound approaches $\sim 1$. More precisely, consider the above protocol (as also given in Fig.~\ref{fig:setup}) in which  random coherent states $\ket{\beta}$ are sent by Bob to its measurement device. If $\rho_{AB}$ is unsteerable, then the minimum value of $\mathcal{W}$ as defined in Eq.~\eqref{eq:witness} is lower bounded by
\begin{align}
    \mathcal{W} \geq \frac{\Delta_\beta^2}{1+\Delta_\beta^2}\;.
\end{align}
Here, $\Delta_{\beta}$ corresponds to the width of the distribution with which $\beta$ is sampled, which we assume to be Gaussian for simplicity, i.e., $P(\beta)=(\pi\sigma^2_\beta)^{-1}{\rm Exp}[-|\beta|^2/\Delta^2_\beta]$. 
This inequality is based on the Bayesian Cram\'er-Rao bound, and it is the standard metrology inequality that can be found in, e.g., Refs.~\cite{genoni2013optimal,morelli2021bayesian,abiuso2021measurement,abiuso2023verification}.

\subsection{Detection Protocol for All Gaussian Steerable States}
\label{app:protocol_viola}

To violate the unsteerable bound Eq.~\eqref{eq:EB_bound_main} by the witness defined in Eq.~\eqref{eq:witness} with all Gaussian steerable states, we consider the protocol detailed in Fig.~\ref{fig:setup}. That is, Bob's station mixes the input coherent state with his share of the source in a 50:50 beam splitter. the outputs of Bob are therefore represented by the observables
\begin{align}
    \hat{b}_1 &=\frac{\hat{x}_B+\hat{x}_\beta}{\sqrt{2}}\;,\\
    \hat{b}_2 &=\frac{-\hat{p}_B+\hat{p}_\beta}{\sqrt{2}}\;.
\end{align}
At the same time, we ask Alice to measure the two possible quadratures $\hat{r}_A=\hat{x}_A,\hat{p}_A$, resulting in  $\mu_r$ and outputs
\begin{align}
    a= & r_{\rm est}(\mu_r)\;,
\end{align}
that is,
\begin{align}
    a= & \langle\hat{x}_B\rangle|_{v_x} \quad \text{if } r=x\;,\\
    a= & \langle\hat{p}_B\rangle|_{v_p} \quad \text{if } r=p\;.
\end{align}
When Alice and Bob follow this protocol, the $x$-term in the witness [Eq.~\eqref{eq:witness}] becomes
\begin{align}
   \left\langle \left( b_1 - \frac{a}{\sqrt{2}}-\beta_x\right)^2\right\rangle\bigg|_{r=x} =
   \left\langle \left( \frac{\hat{x}_B+\hat{x}_\beta}{\sqrt{2}}- \frac{\langle\hat{x}_B\rangle|_{v_x}}{\sqrt{2}}-\beta_x\right)^2\right\rangle 
   =\frac{\langle\hat{x}^2_{\rm err}\rangle}{2} + \frac{{\rm Var} [\hat{x}_\beta]}{2}\;,
\end{align}
and, similarly, the $p$-term is equal to $\frac{\langle\hat{p}^2_{\rm err}\rangle}{2} + \frac{{\rm Var} [\hat{p}_\beta]}{2}$. In our units, the variance  of the quadratures on coherent states is equal to ${\rm Var}[\hat{x}_\beta]={\rm Var}[\hat{p}_\beta]=\frac{1}{2}$. From this, it follows that the entire witness, in the case of the protocol detailed in Fig.~\ref{fig:setup}, becomes
\begin{align}
    \mathcal{W}_{\rm protocol}=\frac{1}{2}+\frac{\langle \hat{x}_{\rm err}^2\rangle + \langle \hat{p}_{\rm err}^2\rangle}{2}.
\end{align}
Consequently, we conclude that
if $\rho_{AB}$ is steerable and such that $\langle \hat{x}_{\rm err}^2\rangle + \langle \hat{p}_{\rm err}^2\rangle<1$, then we have
\begin{align}
    \mathcal{W}_{\rm protocol}< 1\;,
\end{align}
which violates the inequality in Eq.~\eqref{eq:EB_bound_main} 
when $\Delta_\beta$ is large enough. As an example, for a pure two-mode squeezed state (as discussed in Appendix~\ref{sec:TMSS}), we have $\mathcal{W}_{\rm protocol}=\frac{1}{2}+\frac{1}{2\cosh{2\gamma}}$.

\subsubsection{From product to sum}
The final observation we need, in order to prove that all Gaussian steerable states can be detected by $\mathcal{W}_{\rm protocol}$, is the following.
From Reid~\cite{reid1989demonstration} and later works~\cite{wiseman2007steering}, we know that all steerable Gaussian states violate the inequality in Eq.~\eqref{eq:uncert_prod_main} [which is also Eq.~\eqref{eq:uncert_prod}]. This does not automatically guarantee that all steerable states violate also the inequality in Eq.~\eqref{eq:uncert_sum_main} [i.e., Eq.~\eqref{eq:uncert_sum}]. However, when we are given
\begin{align}
 \langle \hat{x}_{\rm err}^2\rangle \langle \hat{p}_{\rm err}^2\rangle <\frac{1}{4},    
\end{align}
then there always exists a Gaussian squeezing 
\begin{align}
    \hat{x}\rightarrow \hat{x}'&=\kappa \hat{x}\;,\\
    \hat{p}\rightarrow \hat{p}'&=\frac{1}{\kappa}\hat{p}\;
\end{align}
such that
\begin{align}
    \langle \hat{x}'^2_{\rm err}\rangle = \langle \hat{p}'^2_{\rm err}\rangle <\frac{1}{2}
\end{align}
and, therefore,
\begin{align}
    \langle \hat{x}'^2_{\rm err}\rangle + \langle \hat{p}'^2_{\rm err}\rangle < 1\;.
\end{align}
The squeezing operation can be implemented locally by Alice and Bob, and can be thought of as part of the measurement apparatus (before the action of the beam splitter).

\end{document}